\documentclass[11pt]{article}
\usepackage{a4,geometry}

\usepackage{graphicx}
\usepackage{amsmath,amssymb,amsthm,mathtools}
\usepackage{paralist}
\usepackage{bm}
\usepackage{bbm}
\usepackage{url}
\usepackage{fullpage, prettyref}
\usepackage{boxedminipage}
\usepackage{wrapfig}
\usepackage{ifthen}
\usepackage{color}
\usepackage{xcolor}
\usepackage{framed}
\usepackage{algorithm}
\usepackage[noend]{algpseudocode}
\usepackage[pagebackref,letterpaper=true,colorlinks=true,pdfpagemode=none,urlcolor=blue,linkcolor=blue,citecolor=violet,pdfstartview=FitH]{hyperref}
\usepackage[nameinlink]{cleveref}
\usepackage{esvect}
\usepackage{mdframed}
\usepackage{thmtools}
\usepackage{thm-restate}
\usepackage{caption}
\usepackage{subcaption}

\usepackage{tikz} 
\usetikzlibrary{positioning}
\usetikzlibrary{calc}
\usetikzlibrary{backgrounds}
\usetikzlibrary{arrows, decorations.pathmorphing}
\usetikzlibrary{decorations.pathreplacing}
\usepackage{standalone}

\newcommand{\reachcycle}[1]{reachable $#1$-cycle}
\newcommand{\reachcycleempty}{reachable cycle}
\newcommand{\reachsimplex}[1]{reachable $#1$-simplex}
\newcommand{\ord}{\pi}

\newcommand{\licl}{\text{LICL}}

\newcommand{\connector}[1]{$#1$-connector gadget}

\newcommand{\setH}[1]{\mathcal{H}_{#1}}
\newcommand{\setHk}[2]{\mathcal{H}_{#1}^{#2}}

\newcommand{\Hom}[2]{\mathrm{Hom}_{#2}(#1)}
\newcommand{\Homr}[2]{\mathrm{rHom}_{#2}(#1)}
\newcommand{\Sub}[2]{\mathrm{Sub}_{#2}(#1)}

\newcommand{\ColHom}[2]{\mathrm{\text{Col-}Hom}_{#2}(#1)}

\newcommand{\img}{Img}
\newcommand{\degen}{\kappa}

\newcommand{\ldegen}[1]{\kappa_{#1}}
\newcommand{\dagtw}{DAG-treewidth}
\newcommand{\dagtree}{DAG-tree decomposition}
\newcommand{\Reach}{Reach}
\newcommand{\dtw}{\tau}
\newcommand{\ldtw}[1]{\tau_{#1}}

\newcommand{\contract}{\cQ^c}
\newcommand{\Skeleton}[2]{Sk_{#1}(#2)}

\newcommand{\ext}[3]{ext(#2,#3;#1)}
\newcommand{\down}{\Gamma}

\newcommand{\ignore}[1]{}

\newcommand{\cB}{\mathcal{B}}

\newcommand{\cE}{{\cal E}}

\newcommand{\cH}{{\cal H}}

\newcommand{\cQ}{\mathcal{Q}}

\newcommand{\cT}{{\cal T}}

\newcommand{\Z}{{\mathbb Z}}

\newcommand{\poly}{\mathrm{poly}}

\newcommand{\NN}{\mathbb{N}}

\declaretheorem[name=Theorem, numberwithin=section]{theorem}
\declaretheorem[name=Lemma, numberlike=theorem]{lemma}
\declaretheorem[name=Claim, numberlike=theorem]{claim}

\declaretheorem[name=Conjecture, numberlike=theorem]{conjecture}

\declaretheorem[name=Definition, numberlike=theorem]{definition}

\declaretheorem[name=Remark, numberlike=theorem]{remark}

\newcommand{\Sec}[1]{\S \ref{sec:#1}} 
\newcommand{\Eqn}[1]{\hyperref[eq:#1]{(\ref*{eq:#1})}} 
\newcommand{\Fig}[1]{{Fig.\,\ref{fig:#1}}} 
\newcommand{\Tab}[1]{\hyperref[tab:#1]{Tab.\,\ref*{tab:#1}}} 
\newcommand{\Table}[1]{\hyperref[tab:#1]{Table\,\ref*{tab:#1}}} 
\newcommand{\Thm}[1]{\hyperref[thm:#1]{Theorem\,\ref*{thm:#1}}} 
\newcommand{\Fact}[1]{\hyperref[fact:#1]{Fact\,\ref*{fact:#1}}} 
\newcommand{\Lem}[1]{\hyperref[lem:#1]{Lemma\,\ref*{lem:#1}}} 
\newcommand{\Prop}[1]{\hyperref[prop:#1]{Prop.~\ref*{prop:#1}}} 
\newcommand{\Cor}[1]{\hyperref[cor:#1]{Corollary~\ref*{cor:#1}}} 
\newcommand{\Conj}[1]{\hyperref[conj:#1]{Conjecture~\ref*{conj:#1}}} 
\newcommand{\Def}[1]{\hyperref[def:#1]{Definition~\ref*{def:#1}}} 
\newcommand{\Alg}[1]{\hyperref[alg:#1]{Alg.~\ref*{alg:#1}}} 
\newcommand{\Clm}[1]{\hyperref[clm:#1]{Claim~\ref*{clm:#1}}} 
\newcommand{\Obs}[1]{\hyperref[obs:#1]{Observation~\ref*{obs:#1}}} 
\newcommand{\Rem}[1]{\hyperref[rem:#1]{Remark~\ref*{rem:#1}}} 
\newcommand{\Con}[1]{\hyperref[con:#1]{Construction~\ref*{con:#1}}} 
\newcommand{\Step}[1]{\hyperref[step:#1]{Step~\ref*{step:#1}}} 
\newcommand{\Assumption}[1]{\hyperref[assm:#1]{Assumption\,\ref*{assm:#1}}} 

\newcommand{\lparam}{l}

\title{Near-linear time subhypergraph counting in bounded degeneracy hypergraphs}
\date{}
\author{Daniel Paul-Pena\thanks{Authors supported by NSF CCF-1740850, DMS-2023495, and CCF-1839317.}\\
	University of California, Santa Cruz\\
	{\tt dpaulpen@ucsc.edu}
	\and C. Seshadhri\footnotemark[1]\\
University of California, Santa Cruz\\
{\tt sesh@ucsc.edu}}

\begin{document}

\maketitle

\begin{abstract}
Counting small patterns in a large dataset is a fundamental algorithmic task. The most common version
of this task is subgraph/homomorphism counting, wherein we count the number of occurrences of a small
pattern graph $H$ in an input graph $G$. The study of this problem is a field in and of itself. 
Recently, both in theory and practice, there has been an interest in \emph{hypergraph} algorithms,
where $G = (V,E)$ is a hypergraph. One can view $G$ as a set system where hyperedges
are subsets of the universe $V$.

Counting patterns $H$ in hypergraphs is less studied, although there
are many applications in network science and database algorithms. Inspired
by advances in the graph literature, we study when linear time algorithms are possible.

We focus on input hypergraphs $G$ that have bounded \emph{degeneracy}, a well-studied concept
for graph algorithms. We give a spectrum of definitions for hypergraph degeneracy that cover
all existing notions. For each such definition, we give a precise characterization
of the patterns $H$ that can be counted in (near) linear time. Specifically, we discover
a set of ``obstruction patterns". If $H$ does not contain an obstruction, then the number
of $H$-subhypergraphs can be counted exactly in $O(n\log n)$ time (where $n$ is the number
of vertices in $G$). If $H$ contains an obstruction, then (assuming hypergraph
variants of fine-grained complexity conjectures), 
there is a constant $\gamma > 0$, such that there is no $o(n^{1+\gamma})$ time algorithm for counting $H$-subhypergraphs.
These sets of obstructions can be defined for all notions of hypergraph degeneracy.
\end{abstract}

\newpage

\section{Introduction}

Subgraph/homomorphism counting in graphs is a widely studied problem in theory and practice~\cite{Lo67, DiSeTh02, FlGr04, DaJo04, Lo12, AhNeRo+15, CuDeMa17, PiSeVi17, SeTi19, CuNe25, DoMaWe25}. The design of algorithms for
homomorphism counting in graphs has gone on for decades, and remains an active area of 
research~\cite{ChMe77,BrWi99,DrRi10,BoChLo+06,PiSeVi17,DeRoWe19,PaSe20, Br19,Br21, BrRo22,PaSe24, PaSe25, PaSe25b}. 
Let $n$ denote the number of vertices in $G$ and $k$ be the vertices in $H$. There is
a trivial brute force $O(n^k)$ algorithm, and so the primary focus of theoretical research
is to beat this bound. Given the practical
importance of this problem, there have been advances in designing (near) linear
time algorithms \cite{BePaSe20,BePaSe21,BeGiLe+22,PaSe25}.

Recent work in network science and the theory of algorithms has focused on \emph{hypergraphs} \cite{VeBeKl20,VeBeKl22,AnCoPo+23,CaDeFa+23,LeBuEl+25}.
A hypergraph $G = (V(G), E(G))$ is a set system, where each $e \in E(G)$ is a subset of $V(G)$
of size at least two. Many natural real-world graphs are
derived from hypergraphs. For example, coauthor/coactor networks are obtained from the hypergraph where authors are vertices, and papers are hyperedges. Hyperedges are replaced with cliques to get a coauthor network. Hence, hypergraphs are considered to be ``original" data source,
and practitioners suggest that they should be analyzed directly.

The problem of pattern counting in hypergraphs is less understood theoretically, even though
it captures many important algorithmic tasks~\cite{HwTaTi+08,YuTaWa12,BeAbSc+18,HuQiSh+10,AmVeBe20}. For example, the entire field of Conjunctive Query
evaluation in database theory can be viewed as pattern counting in hypergraphs \cite{DeKo21,AhMiOl+23}. Any database
can be represented as a (colored) hypergraph, and the query is a pattern hypergraph $H$.
We note that recent work by Bressan-Brinkmann-Dell-Roth-Wellnitz has studied FPT algorithms
for subhypergraph counting~\cite{BrBrDe+25}.

Formally, given a pattern hypergraph $H = (V(H), E(H))$, an 
$H$-homomorphism is a map $f:V(H) \to V(G)$ that preserves hyperedges. Abusing notation, given
$e = \{v_1, v_2, \ldots\}$, let $f(e) = \{f(v_1), f(v_2), \ldots\}$.
So, $\forall e \in E(H)$, $f(e) \in E(G)$. If $f$ is an injection (so distinct vertices of $H$ are mapped
to distinct vertices of $G$), this map corresponds to some subhypergraph \footnote{Such a map is an embedding. Each subhypergraph corresponds to $|Aut(H)|$ different embeddings, where $Aut(H)$ is the number of automorphisms of $H$ (see \cite{BrBrDe+25}).}.
We use $\Hom{G}{H}$ (resp. $\Sub{G}{H}$) to denote the count of the distinct $H$-homomorphisms (resp. $H$-subhypergraphs). 
The focus of this work is:

\begin{center} {\em Are there characterizations of $H$ and classes of $G$ where $\Hom{G}{H}$ can 
    be computed in near-linear time?}
\end{center}

\subsection{Main result} \label{sec:result}

A natural starting point for near-linear time algorithms for hypergraph homomorphism counting
is to focus on \emph{bounded degeneracy hypergraphs}\footnote{Technically, we mean hypergraphs belonging to classes with bounded degeneracy.}. 
There is a recent theory of linear time algorithms for bounded degeneracy \emph{graphs}~\cite{Br19,BePaSe20,BePaSe21,BrRo22,BeGiLe+22, PaSe24, PaSe25}.
This is a rich class containing all proper minor-closed families, bounded expansion families, and preferential attachment graphs~\cite{Se23}.
Practical subgraph counting algorithms typically use algorithmic techniques
for bounded degeneracy graphs~\cite{AhNeRo+15,JhSePi15,PiSeVi17,OrBr17,JaSe17,PaSe20}. 

Recent works in database theory have focused in studying databases with "degree constraints" \cite{KhNgHu+16}, and this theory has lead to cutting edge algorithms for query evaluation \cite{panda}. A database with degree constraints can be represented as a bounded degeneracy hypergraph. A recent result by Deeds and Merkl \cite{DeMe25} introduced partition constraints for databases, which are essentially equivalent to breaking databases/tables into hypergraphs with bounded degeneracy.

There are numerous
possible definitions for the degeneracy of hypergraphs. Our results provide dichotomy theorems for
near-linear homomorphism and subhypergraph counting algorithms for all these definitions. 

As usual, the input hypergraph is $G = (V(G), E(G))$ with 
$n$ vertices and $m$ hyperedges. We consider $H$ to be constant sized and independent of the input, and fold dependencies on $H$ into $O(\cdot)$ or $poly(\cdot)$. The \emph{arity} of a hyperedge is its size. We allow $G$
to have hyperedges of different arity. The \emph{rank} of a hypergraph is the maximum arity among all its hyperedges. We assume that the rank of $G$ is bounded, and again suppress
dependencies on the rank in $O(\cdot)$/$\poly(\cdot)$ notation.\footnote{Note that hyperedges of unbounded arity can not be mapped by any hyperedge of $H$, and therefore do not affect the homomorphism/subhypergraph counts. If given a hypergraph of unbounded rank, one can filter out the high arity hyperedges in $O(m)$ time to get a bounded rank hypergraph.}

Hypergraphs introduce complications when defining \emph{induced} subhypergraphs,
which this definition lays out. 

\begin{definition} \label{def:ind} Let $\lparam \in \Z^+ \cup \{\infty\}$. For any subset $S \subseteq V(G)$,
the induced $\lparam$-trimmed subhypergraph on $S$ is the set of hyperedges
$\{e \cap S \ | \ e \in E(G), |e \setminus S| \leq \lparam\} \setminus \{\emptyset\}$.

    Conventionally, if $\lparam$ is not specified, then it is set to $\infty$. (Depending on context,
    arity 1 hyperedges might be discarded.)
\end{definition}

\begin{figure}
	\centering
	\includegraphics[width=0.9\linewidth]{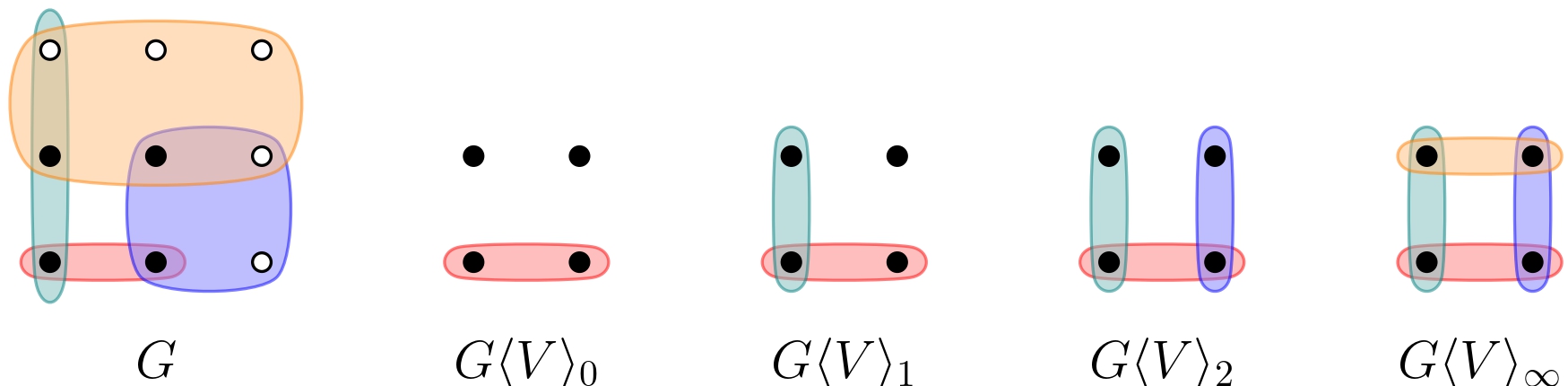}
	\caption{An example of a graph $G$. Let $V$ be the set of black vertices. $G \langle V \rangle_{\lparam}$ represents the $l$-trimmed subhypergraph of $G$ induced by the vertices in $V$. A hyperedge is only included in the subhypergraph when it has at most $l$ vertices outside $V$.}
	\label{fig:trimmed}
\end{figure}

We note that the term ``trimming'' comes from~\cite{BrBrDe+25}, though they only define the $\lparam = \infty$ version. See \Fig{trimmed} for an example of different induced $\lparam$-trimmed subhypergraphs.
When $\lparam = 0$, this is the simplest (standard) notion of induced subhypergraph. Only hyperedges contained
in $S$ are considered. In some settings, hypergraphs are thought of as simplicial
complexes wherein all subsets of a hyperedge are also part of $G$. This corresponds to $\lparam = \infty$, where
an induced subhypergraph is formed by taking the intersection of every hyperedge with $S$
(technically, we only need to set $\lparam$ to be the rank of $G$).
For intermediate values of $\lparam$, we only include the intersection $e \cap S$ if at most $\lparam$
vertices of $e$ are outside $S$. This concept is related to nuanced notions of hypergraph cuts \cite{VeBeKl20,VeBeKl22}. 

We now define the hypergraph degeneracy. The degree of a vertex is the number of hyperedges in which it participates.

\begin{definition} \label{def:degen} The $\lparam$-degeneracy of $G$, denoted $\ldegen{\lparam}(G)$, is the
minimum value $\degen$ such that every induced $\lparam$-trimmed subhypergraph of $G$ has minimum degree of at most $\degen$.
\end{definition}

Note that for $\lparam < \lparam'$, $\ldegen{\lparam}(G) \leq \ldegen{\lparam'}(G)$. So, having bounded $0$-degeneracy
is a weaker condition than having bounded $\lparam$-degeneracy (for $\lparam > 0$). The $0$-degeneracy corresponds to the standard definition of degeneracy in hypergraphs \cite{KoRo06}, while the $\infty$-degeneracy is polynomially equivalent to the degeneracy of the clique completion of $G$ (assuming bounded rank), also called skeletal degeneracy \cite{FoSaSi+23}.

Our aim is to characterize $H$ where $\Hom{G}{H}$ and $\Sub{G}{H}$ can be counted in near-linear time for bounded $\lparam$-degeneracy inputs $G$.
Towards this characterization, we define a series of obstruction hypergraphs.

\begin{definition} \label{def:obstruct} A hypergraph $H$ is an \emph{$\lparam$-obstruction} if the following hold. For some $k \geq 3$ there is a \emph{core} $C \subseteq V(H)$ of size $|C| = k$.
Let the connected components of the induced $\infty$-trimmed subhypergraph on $V(H) \setminus C$
be (the vertex sets) $D_1, D_2, \ldots$. For each $D_i$, let $E_i$ be the set of hyperedges
of $E(H)$ that are completely contained in $C \cup D_i$, that is $E_i = \{e \in E(H) : e \subseteq C \cup D_i, |e|>1 \}$. 

The core $C$ and $E_1, E_2, \ldots$ satisfy the following conditions:
\begin{asparaenum}
	\item The $\lparam$-trimmed subhypergraph induced by the core $C$ is empty\footnote{It is allowed to include hyperedges of arity $1$, as they do not affect the connectivity of the core.}.
	\item There is no $E_j$ that contains all of $C$. $\forall E_j,  C \not\subseteq \bigcup_{e \in E_j} e$.
    \item For each $c \in C$, $E_i$ contains\footnote{Technically, we number the vertex sets $D_1, D_2, \ldots$ so that this condition holds.} the vertices $C \setminus \{c_i\}$ and \emph{does not} contain $c_i$.
\end{asparaenum}
\smallskip
The set of $\lparam$-obstructions is denoted by $\setH{\lparam}$. 
A hypergraph $H'$ is called $\setH{\lparam}$ ITS (induced trimmed subhypergraph) free if it does not contain
any member of $\setH{\lparam}$ as an induced $\infty$-trimmed subhypergraph.
\end{definition}
\begin{figure}
	\centering
	\includegraphics[width=0.75\linewidth]{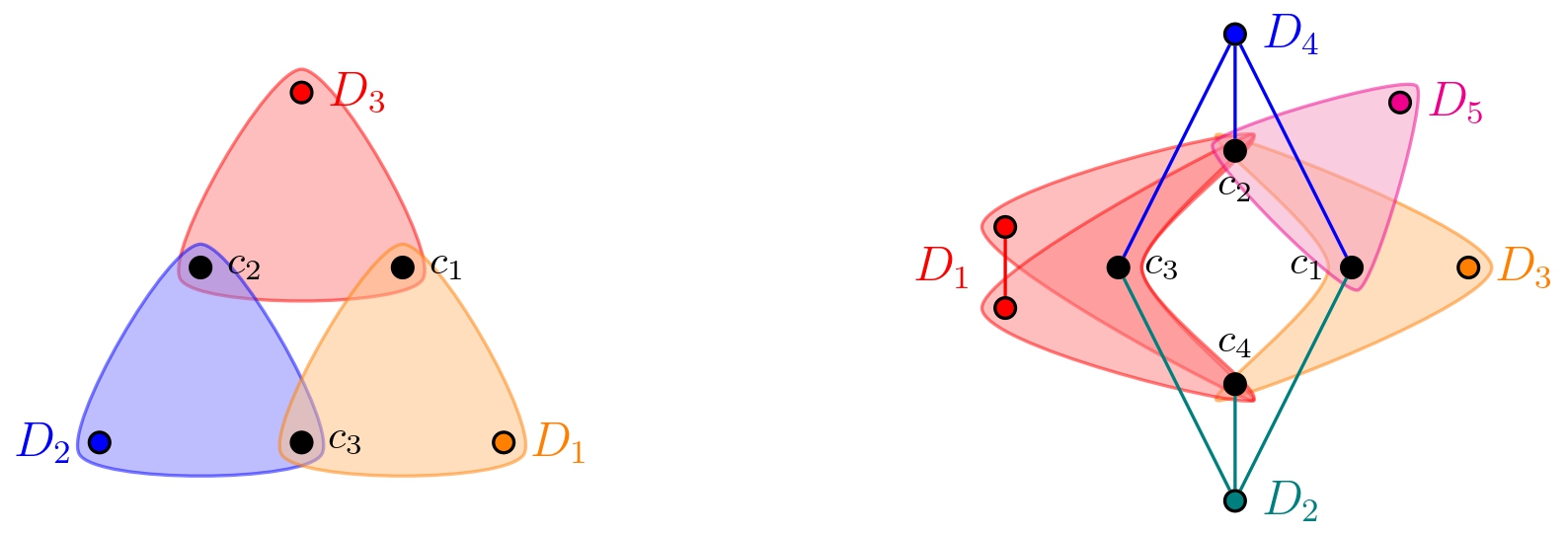}
	\caption{Two examples of obstructions in $\setH{0}$. 
		The example on the left has a core with $k=3$. Removing the vertices from the core disconnects the hypergraph into $3$ connected components (single vertices $D_1-D_3$). The sets $E_1-E_3$ correspond to the hyperedges containing each $D_i$.
		The example on	the right has a core with $k=4$. Removing the core we end up with $5$ connected components. Note that we have an extra component $D_5$.
	}
	\label{fig:obstruction}
\end{figure}
\begin{remark}
	As mentioned, the sets $D_1, D_2,...$ are the connected components of the induced $\infty$-trimmed subhypergraph on $V(H) \setminus C$. From condition $1$ we have that the core is empty, that is, every hyperedge of arity more than $1$ must include some vertex outside the core. Note that from the definition, no edge can contain vertices in two different sets $D_i,D_j$. Therefore every hyperedge will belong to an unique set $E_i$. 
	
	Condition $2$ specifies that no set $E_i$ contains all the vertices in the core, while condition $3$ says that for each vertex $c_i$ in the core, there is at least a set that contains $C\setminus \{c_i\}$. We reorder the sets $D_i$ and $E_i$ such that each $E_i$ corresponds to $c_i$. Note that there might be more sets $E_i$ than elements in the core.
\end{remark}

We give more intuition on these conditions in \Sec{perspective}.
This is a technical definition, but we can prove that \Def{obstruct} are the only obstructions for near-linear time algorithms.

For the hardness, we state a generalization of the Triangle Detection Conjecture \cite{AbWi14}, which is the fundamental assumption used for linear time hardness for subgraph counting \cite{BePaSe21,BeGiLe+22, PaSe24, PaSe25}.
This conjecture states that there is no near-linear time algorithm for finding triangles in an arbitrary graph.
The best known algorithm takes $O(m^{1.41\ldots})$ time \cite{AlYuZw97}.
A $k$-simplex is a hypergraph with $k+1$ vertices and $k+1$ hyperedges of arity $k$, where each hyperedge includes a different subset of $k$ vertices. The triangle is the $2$-simplex, while the tetrahedron is the $3$-simplex. Detecting simplices is the hypergraph generalization of triangle detection. 
No algorithm is known for detecting simplices of any arity in linear time. 

\begin{conjecture}[Simplex detection conjecture] \label{conj:simplex}
	For every $k\geq 2$, there exists a $\gamma>0$ such that any (even randomized) algorithm to decide whether a hypergraph with $n$ vertices and $m$ hyperedges contains a $k$-simplex requires $\Omega(m^{1+\gamma})$ time in expectation. \footnote{The $(k+1,k)$-Hyperclique Hypothesis presented in \cite{LiVaWi18} conjectures that any algorithm for finding a $k$-simplex in a $k$-regular hypergraph requires $n^{k+1-o(1)}$ time. In the dense case, it is a stronger assumption than our conjecture.}
\end{conjecture}

\begin{restatable}[name=Main Theorem]{theorem}{main}
\label{thm:main}
    Let $\lparam \in \Z^+ \cup \{\infty\}$ and $H$ be a hypergraph.
\begin{itemize}
	\item Suppose $H$ is $\setH{\lparam}$ ITS free. Then, there is an algorithm that computes $\Hom{G}{H}$ in time
$\poly(\ldegen{\lparam}) n\log n$. \footnote{Technically, the result of this theorem is fixed-parameter near-linear time. However, this implies a near-linear time algorithm for hypergraph classes with bounded degeneracy.}
	\item Otherwise, assuming the Simplex Detection Conjecture, there exists an absolute constant $\gamma > 0$,
such that: for all functions $f:\NN \to \NN$, any algorithm that computes $\Hom{G}{H}$ requires $f(\ldegen{\lparam}) n^{1+\gamma}$ time.
\end{itemize}
\end{restatable}

Thus, for bounded $\ldegen{l}$ for any value of $l$, we have a precise characterization of when $\Hom{G}{H}$ can be computed in near-linear time.

We can extend the hardness characterization to subhypergraph counting. The quotient set of $H$  is the set of hypergraphs obtained by merging sets of vertices that form a partition of $H$ (see \Def{quotient} for a formal definition). Like in the case of graphs, we can use the homomorphism basis to lift the hardness from counting homomorphisms of the quotient set of $H$ to counting subhypergraphs of $H$.

\begin{theorem}[Characterization for subhypergraphs] \label{thm:sub}
    Let $\lparam \in \mathbb{Z}^+\cup \{\infty\}$ and $H$ be a hypergraph.
    \begin{itemize}
    	\item If every hypergraph in the quotient set of $H$ is $\setH{l}$ ITS free. Then, there is an algorithm that computes $\Sub{G}{H}$ in time $poly(\ldegen{l})n \log{n}$.
    	\item Otherwise, assuming the simplex detection conjecture, there exists an absolute constant $\gamma > 0$,
    	such that: for all functions $f:\NN \to \NN$, any algorithm that computes $\Sub{G}{H}$ requires $f(\ldegen{\lparam}) n^{1+\gamma}$ time.
    \end{itemize}
\end{theorem}

\subsection{Perspectives on \Def{obstruct} and \Thm{main} } \label{sec:perspective}

We explain why all the technicalities of \Def{obstruct} are necessary. This provides useful
intuition and also demonstrates why subhypergraph counting is technically challenging. We consider
\Def{obstruct} as our main discovery, which illustrates why subhypergraph counting is not
a simple generalization of subgraph counting.

We give examples showing that removing any aspect of \Def{obstruct} leads to patterns that can be counted in linear time.
Let us focus on $\lparam = 0$ for ease of presentation.
\Fig{obstruction} has examples of patterns in $\setH{0}$.
The edge sets $E_i$ in \Def{obstruct} are called ``connectors''. For convenience, we use ``hardness''
for near-linear time hardness. When we say ``counting'', we refer to homomorphism counting. 

For context, we recall the characterization of hard pattern graphs for graph homomorphism counting \cite{BePaSe21,BeGiLe+22}.
If the longest induced cycle length (LICL) of $H$ is at least $6$, it is hard; otherwise,
there is a near-linear time algorithm to count $\Hom{G}{H}$. The difficulty of discovering
\Def{obstruct} and the many technical conditions for hypergraph counting should be contrasted with the easy categorization
of the graph case. 

\begin{figure}
	\centering
	\includegraphics[width=0.9\linewidth]{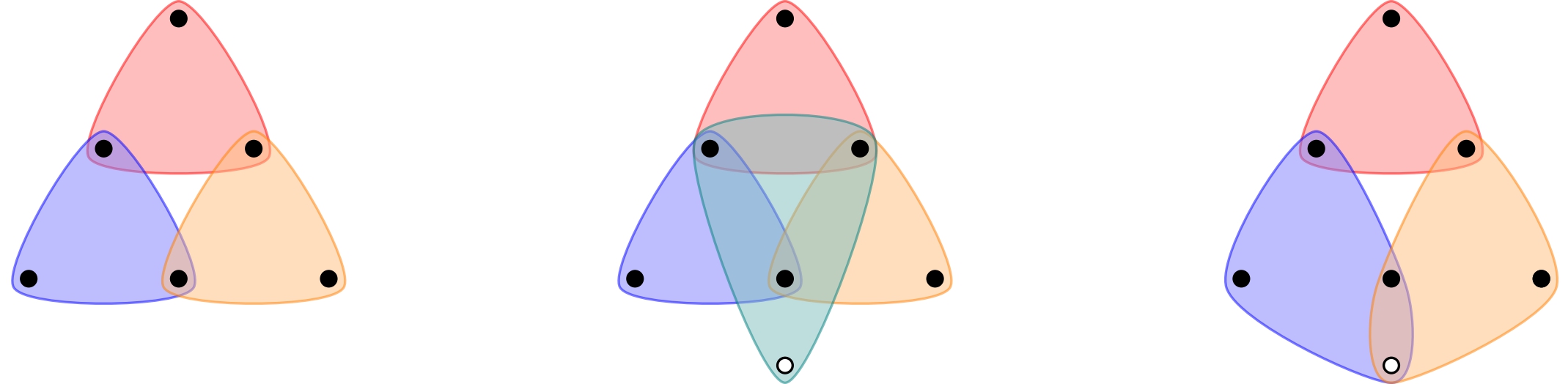}
	\caption{Left: An obstruction in $\setH{0}$. Center: A pattern that contains the obstruction as a ($0$-trimmed) subhypergraph when induced by the black vertices, but it does not contain any obstruction as a $\infty$-trimmed subhypergraph. Therefore, it can be counted efficiently. Right: A pattern that contains the obstruction as the $\infty$-trimmed hypergraph induced by the black vertices. Note that the induced $0$-trimmed subhypergraph only contains the red hyperedge. This pattern cannot be counted efficiently.}
	\label{fig:trimmed_vs_induced}
\end{figure}

{\bf Containing $\setH{0}$ as induced $\infty$-trimmed subhypergraphs:} \Thm{main}
states that the hard patterns contain a member of $\setH{0}$ as induced $\infty$-trimmed subhypergraphs.
We cannot replace that with induced (say) $0$-trimmed subhypergraphs, which is the typical
definition of induced hypergraphs. 
In fact, we can find both patterns that contain an obstruction in $\setH{0}$ as an induced ($0$-trimmed) subhypergraph that are easy to count, and patterns that do not contain an obstruction in $\setH{0}$ as an induced subhypergraph that are hard to count.
For example, consider \Fig{trimmed_vs_induced}.
On the left, we have a pattern $H$ in $\setH{0}$. In the center, we modify $H$ and obtaining a pattern $H'$ that contains $H$ as an induced $0$-trimmed (standard) subhypergraph but not as an induced $\infty$-trimmed subhypergraph. Our algorithm can count $\Hom{G}{H'}$ in near-linear time. Finally, on the right, we modify it again, obtaining $H''$, which contains $H$ as an induced $\infty$-trimmed subhypergraph but not as an induced $0$-trimmed subhypergraph. We can prove hardness for counting this pattern.

\begin{figure}
	\centering
	\includegraphics[width=0.75\linewidth]{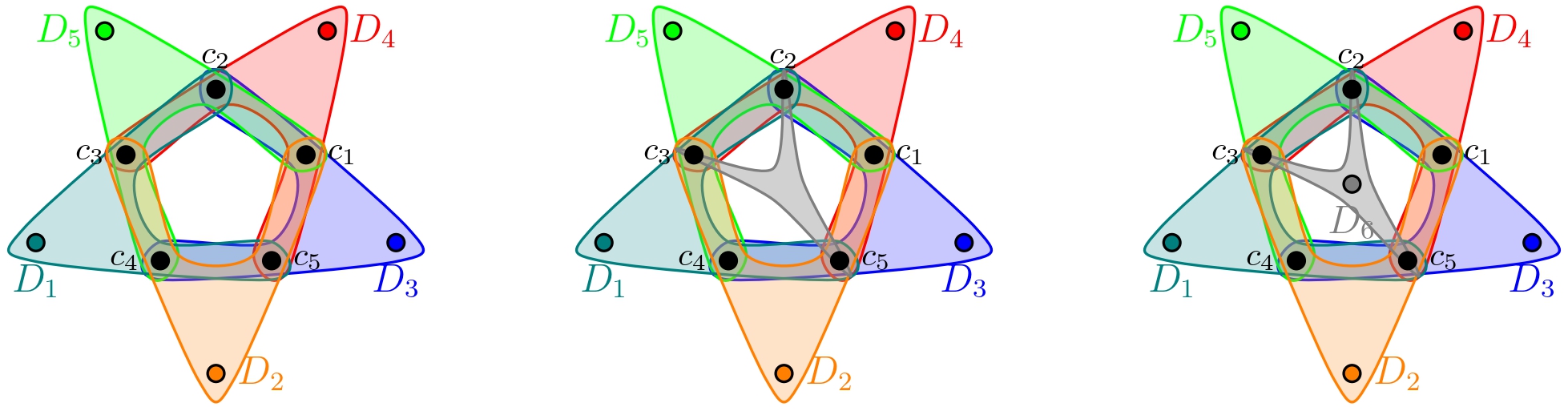}
	\caption{An example of an obstruction with a core of size $5$. Adding the gray hyperedge no longer results in an obstruction. Adding an external vertex to the gray hyperedge creates a new obstruction.}
	\label{fig:simplex5}
\end{figure}

{\bf The empty core:} The first condition of \Def{obstruct} requires that the core $C$ induces
no $\lparam$-trimmed hyperedges. As an example, we give an obstruction $H \in \setH{0}$
with $k=5$ (\Fig{simplex5}). Note that each hyperedge contains five vertices; four in the core, and one vertex outside. If we add any hyperedge to the core,
the resulting $H'$ can be counted in near-linear time. Our algorithm can exploit this extra connectivity
to efficiently count $H'$. Adding an external vertex to that hyperedge results in a pattern that is again hard to count, as the hyperedge just forms a new connector, and the induced $0$-trimmed subhypergraph of the core is again empty. 

This core condition is the only place where $\lparam$ occurs. Refer to \Fig{l_difference}. We can construct a pattern $H \in \setH{1}$,
so it is hard to count for both $0$-degeneracy and $1$-degeneracy bounded input graphs. We modify
it so that the $1$-trimmed hypergraph induced by the core is non-empty, but the corresponding
$0$-trimmed subhypergraph is empty. So, it remains hard for bounded $0$-degeneracy inputs, but can
be counted in near-linear time for bounded $1$-degeneracy inputs.

{\bf No $E_j$ contains the core:} We give an obstruction $H \in \setH{0}$ in \Fig{connector_all}.
We add one hyperedge to get $H'$ so that, according to \Def{obstruct}, there is a connector of hyperedges $E_j$ that contains the three vertices in the core. We can also see that $H'$ does not contain any obstruction
as an induced trimmed subhypergraph. So $H'$ can be counted by our algorithms in near-linear time.

{\bf $E_i$ contains $C \setminus \{c_i\}$:} Let us explain the last condition of \Def{obstruct}.

We give two examples. First, consider an example of an obstruction $H \in \setH{0}$ with $k=3$. The core has vertices $c_1, c_2, c_3$. As shown in \Fig{nocover}, there is a connector $E_3$ that contains $\{c_1, c_2\}$. We disconnect $D_3$, so no connector spans $\{c_1,c_2\}$. The resulting pattern is no longer an obstruction.

Second, consider now an obstruction in $\setH{0}$ with core $c_1,c_2,c_3,c_4$. As shown in \Fig{nocoverb}, there is a connector $E_4$ that contains $\{c_1,c_2,c_3\}$. We split it into two connectors spanning $\{c_1,c_2\}$ and $\{c_2,c_3\}$. Intuitively, all portions of the core are connected, but $H'$ is no longer an obstruction. 

However, in this case, we can see that the $\infty$-trimmed subhypergraph induced by all the vertices except $c_4$ will actually be an obstruction in $\setH{0}$ with $\{c_1, c_2, c_3\}$ as the core. Note that this was not the case before.

\begin{figure}[t]
	\centering
	\begin{subfigure}[t]{.51\textwidth}
		\centering
		\includegraphics[width=0.8\linewidth]{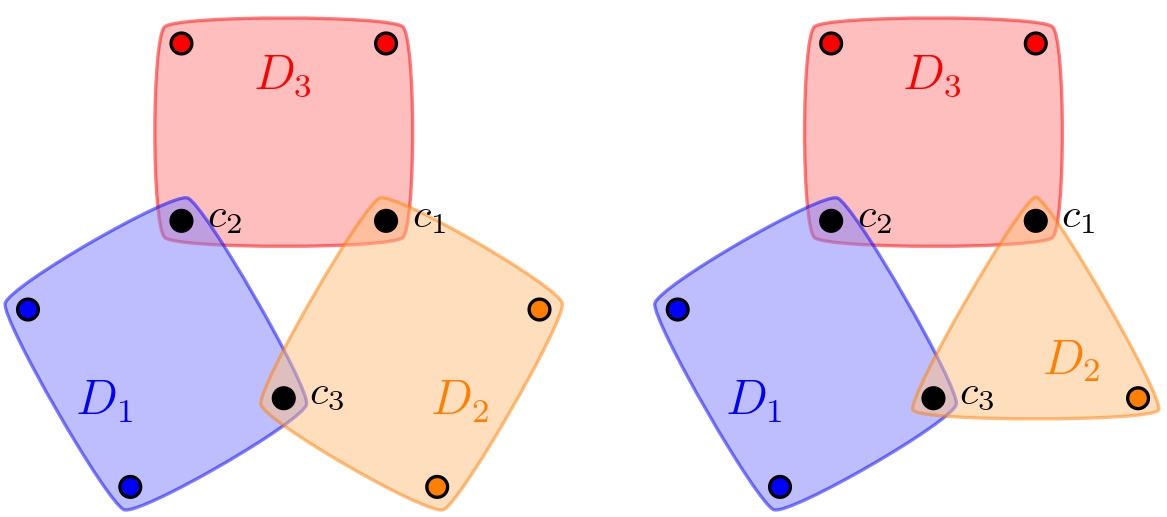}
		\caption{The pattern on the left is an obstruction in both $\setH{0}$ and $\setH{1}$. By contracting the two orange vertices, the $1$-trimmed subhypergraph induced by the core will no longer be empty. It will not be an obstruction in $\setH{1}$.}
		\label{fig:l_difference}
	\end{subfigure}
	\hspace{0.03\textwidth}
	\begin{subfigure}[t]{.43\textwidth}
		\centering
		\includegraphics[width=1\linewidth]{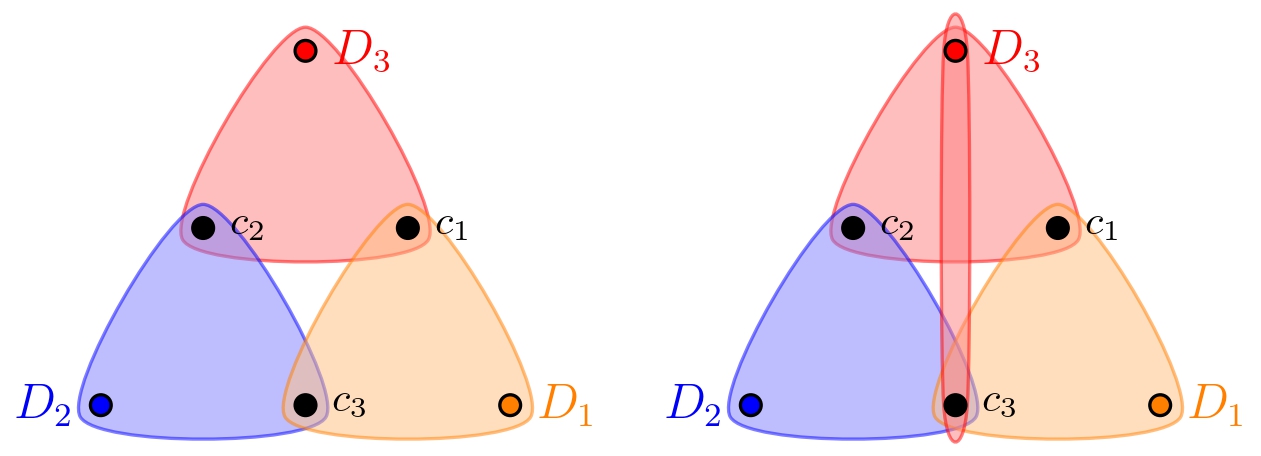}
		\caption{An example of an obstruction in $\setH{0}$. Adding a hyperedge makes $E_3$ contain all the vertices in the core, resulting in a pattern not in $\setH{0}$.}
		\label{fig:connector_all}
	\end{subfigure}%
	\caption{More examples of obstructions.}
	\label{fig:covers}
\end{figure}

{\bf The role of $k$:} It is easiest to understand $k=3$. In this case,
we can think of the connectors as (hyper)paths that connect pairs in the core.
The obstruction looks like an induced hypercycle with at least 6 vertices.
This is exactly the pattern characterization in the graph setting, of having
an induced cycle of length at least 6. For the graph setting, we can prove
that obstructions for larger $k$ contain an induced cycle with at least $6$ vertices.
Even for the hypergraph case, if all hyperedges include at most one vertex
of the core, we can show a similar statement (\Clm{l-infty}).

It becomes complicated when hyperedges of higher arity contain multiple vertices of the core.
Consider the obstructions in \Fig{obstruction} with $k = 4$ and in \Fig{simplex5} with $k=5$.
They do not contain
an obstruction with a smaller $k$, and therefore \Def{obstruct} must consider all values of $k$.

\begin{figure}[t]
	\centering
	\begin{subfigure}[t]{.350\textwidth}
		\centering
		\includegraphics[width=1\linewidth]{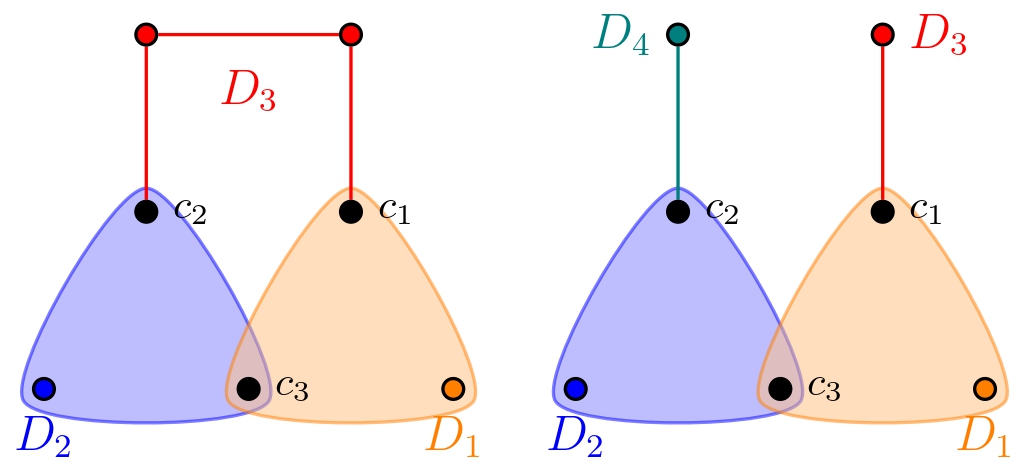}
		\caption{An example of an obstruction in $\setH{0}$. Disconnecting $D_3$ makes the pattern no longer an obstruction, as no connector covers $\{c_1,c_2\}$.}
		\label{fig:nocover}
	\end{subfigure}%
	\hspace{0.03\textwidth}
	\begin{subfigure}[t]{.59\textwidth}
		\centering
		\includegraphics[width=\linewidth]{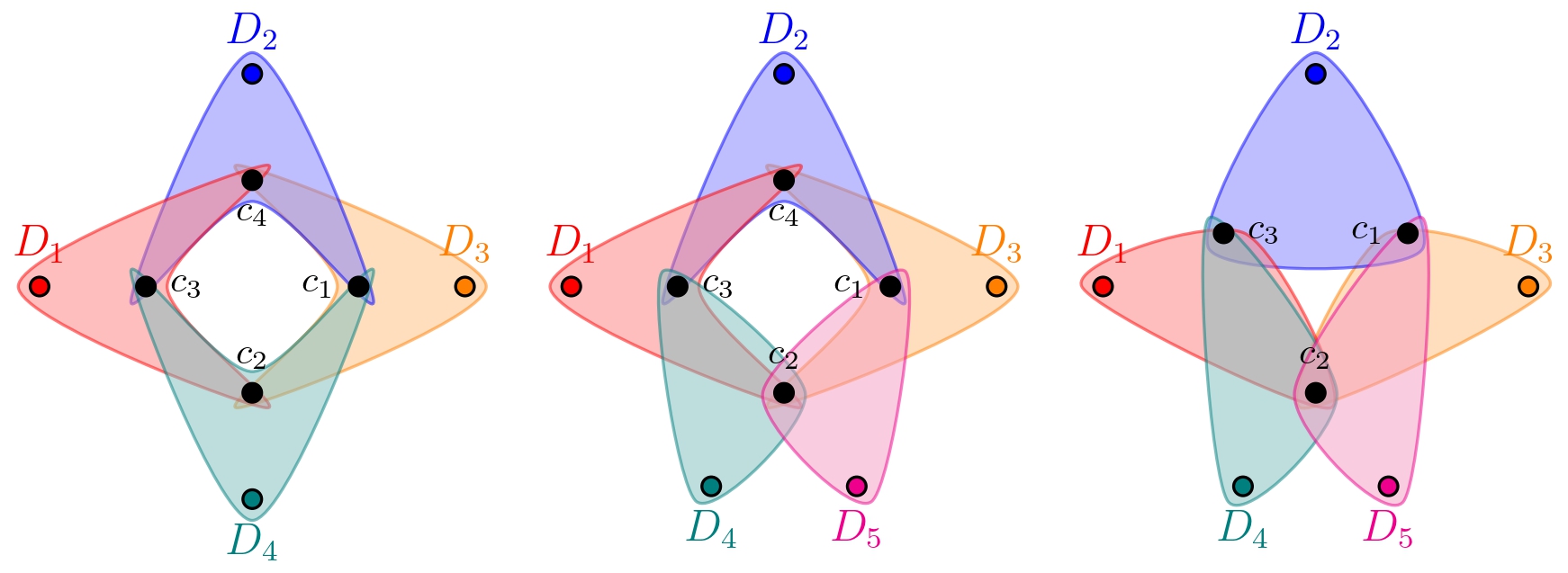}
		\caption{An example of an obstruction $H$ in $\setH{0}$. We obtain $H'$ by splitting the teal hyperedge into two different connectors, $C$ will no longer be a valid core. However, the induced $\infty$-trimmed subhypergraph by $V(H')\setminus \{c_4\}$ (right) actually forms an obstruction.}
		\label{fig:nocoverb}
	\end{subfigure}
	\caption{Examples showcasing the connectivity requirements of the connectors.}
	\label{fig:coversb}
\end{figure}

{\bf The role of $\lparam$:} As mentioned earlier, $\lparam$ only appears
in the first condition of \Def{obstruct}. We find it curious that 
all other notions of induced subhypergraphs either used the $0$-trimmed
or $\infty$-trimmed variants, regardless of the $\lparam$ value.  Moreover, note that for $\lparam < \lparam'$, $\setH{l} \supseteq \setH{l'}$.

When $\lparam = \infty$, the obstruction becomes simpler and defaults to
the graph case. Let the clique completion of a hypergraph be the graph
obtained by replacing each hyperedge with a clique. We can prove the following.

\begin{restatable}{claim}{linfty} \label{clm:l-infty}
	  $H$ is $\setH{\infty}$ ITS free iff $\licl(H^c) < 6$ where $H^c$ is the clique completion of $H$.
\end{restatable}

Where $\licl$ is the length of the longest induced cycle in $H^c$. This is not too difficult to prove, since no hyperedge can contain more than one vertex of the core.
As we discussed earlier, the connectors can be broken down into paths, and it defaults to the $k=3$ setting.

{\bf Why hypergraphs are challenging:} The examples shown highlight the difficulties
in characterizing hard patterns for near-linear time algorithms. In the graph case, the basic
obstruction is a subdivision of a triangle, and the hardness of general triangle counting
neatly matches up with the algorithmic benefits of bounded degeneracy inputs. For hypergraphs,
while the same philosophy holds, the obstructions are much more complicated. At a high level,
it is some embedding of a simplex, but the specific connectivity imposed by \Def{obstruct}
is quite tricky. Obstructions can be quite brittle, in that adding a strategic hyperedge
allows for efficient algorithms. Overall, we believe that this shows the need for more
theory of subhypergraph counting, and the potential for precise dichotomy theorems.

\subsection{Related work} \label{sec:related}

Subgraph and homomorphism counting are vast topics, and we only give a brief
list of papers in this area. We focus on results that are directly relevant
for our work.

Tree decompositions have played a central role in subgraph/homomorphism counting.
Tree decomposition and treewidth were first discovered by Bertele-Brioschi~\cite{BeBr73}
and Halin~\cite{Ha76}. They were also developed in the context of graph minor theory
by Robertson and Seymour~\cite{RoSe83,RoSe84,RoSe86}.
A classic result of D{\'\i}az et al ~\cite{DiSeTh02} gives a $O(2^{k}n^{tw(H)+1})$
algorithm for determining the $H$-homomorphism count,
where $tw(H)$ is the treewidth of $H$. Assuming $\#W[1]$-hardness, Dalmau and Jonsson~\cite{DaJo04} proved
that $\Hom{G}{H}$ is polynomial time solvable iff
$H$ has bounded treewidth. 

The notion of graph \emph{degeneracy} has played an important role
in subgraph counting. The first use was pioneered by Chiba-Nishizeki for clique counting~\cite{ChNi85}.
Since then, it has been used in many theoretical and applied
subgraph counting results~\cite{ChNi85,Ep94,AhNeRo+15,JhSePi15,PiSeVi17,OrBr17,JaSe17,PaSe20}.
The short survey of Seshadhri discusses these connections~\cite{Se23}.

A seminal result of Bressan connected degeneracy to tree decompositions, through
the concept of DAG treewidth~\cite{Br19, Br21}. This concept has been central
in many near-linear time algorithms for bounded degeneracy inputs. 
Bera-Pashanasangi-Seshadhri initiated the theory of linear time homomorphism
counting \cite{BePaSe20}, showing that all patterns with at most $5$ vertices could be counted in linear time. 
Bera-Pashanasangi-Seshadhri and Bera-Gishboliner-Levanzov-Seshadhri then gave a precise characterization,
showing that $H$-homomorphisms can be counted in near-linear time in bounded degeneracy graphs
iff $LICL(H) < 6$ \cite{BePaSe21,BeGiLe+22}. (This assumes fine-grained complexity conjectures.) 

Hypergraph cut functions are a relatively recent area of study, but there has been a surge
of research on this problem in the network science community\cite{VeBeKl20,VeBeKl22}. Subhypergraph counting is not as well studied.
A recent work by Bressan-Brinkmann-Dell-Roth-Wellnitz studied FPT algorithms
for subhypergraph counting~\cite{BrBrDe+25}, showing that the homomorphism basis also extends to the hypergraph setting. The problem of homomorphism counting in hypergraphs has also been studied in hypergraphs of both bounded \cite{DaJo04} and unbounded rank \cite{GrMa14} from a parameterized complexity point of view.

The triangle detection conjecture was first stated by Abboud-Williams for getting
fine-grained complexity results \cite{AbWi14}.   
The best known algorithm for the triangle detection problem uses fast matrix multiplication and runs in time
$O(\min \{n^\omega, m^{{2\omega}/{(\omega+1)}}\}) \approx O(m^{1.41\ldots})$~\cite{AlYuZw97}. If $\omega=2$, this yields a running time of $m^{4/3}$, which
many believe to be the true complexity. 
Any improvement even on this bound would be considered a huge breakthrough in algorithms, and a near-linear time algorithm would completely
upend our understanding of triangle counting.

\section{Technical overview}

In this section, we give a brief exposition of the main technical aspects of the paper along with the main results of each of the sections.

\subsection{Directed acyclic hypergraphs}

Most of the work in homomorphism counting of degenerate graphs requires exploiting orientations of the input graph with low outdegree and reducing the problem to homomorphism counting of directed acyclic graphs. This approach was first started by Chiba-Nishizeki \cite{ChNi85}.

If we wish to do the same in the hypergraph setting, we need to start by defining directed hypergraphs. In the most common definition of a directed hypergraph, each edge is defined by two sets of vertices: the tail and the head \cite{AuLa17}; this directly generalizes directed graphs where, for every directed edge, one vertex is the tail and another the head. 

However, an alternative way of generalizing directed acyclic graphs is by using an ordering of the vertices. We will use the term directed acyclic hypergraph (DAH) to refer to this variant of direct hypergraphs.

\begin{definition}[Directed acyclic hypergraph]
	A directed acyclic hypergraph $\vec{G}$ is a hypergraph $G$ with an ordering function $\ord: V(G)\to [|V(G)|]$.
\end{definition}
 
 In a DAH, every hyperedge will be given by a sorted list of vertices that matches the order of the hypergraph.

\subsection{$\lparam$-skeletons}

In relation to DAH, we define different notions of ``skeletons''. The $\lparam$-skeleton of a DAH $\vec{H}$ is the directed acyclic graph obtained by connecting each of the first $(\lparam+1)$ vertices on each hyperedge of $\vec{G}$ to any vertex that appears after them in the internal ordering of the hyperedge. If $\lparam=0$, this is equivalent to replacing every hyperedge by an ``out-star'' (a star where the direct edges go from the center to the external nodes). If $\lparam = \infty$ (or just an integer at least the rank of $H$ minus $2$), this becomes the orientation of $\vec{H}^c$ that matches the ordering of $\vec{H}$. See \Fig{skeleton} for an example.

\begin{figure}
	\centering
	\includegraphics[width=1\linewidth]{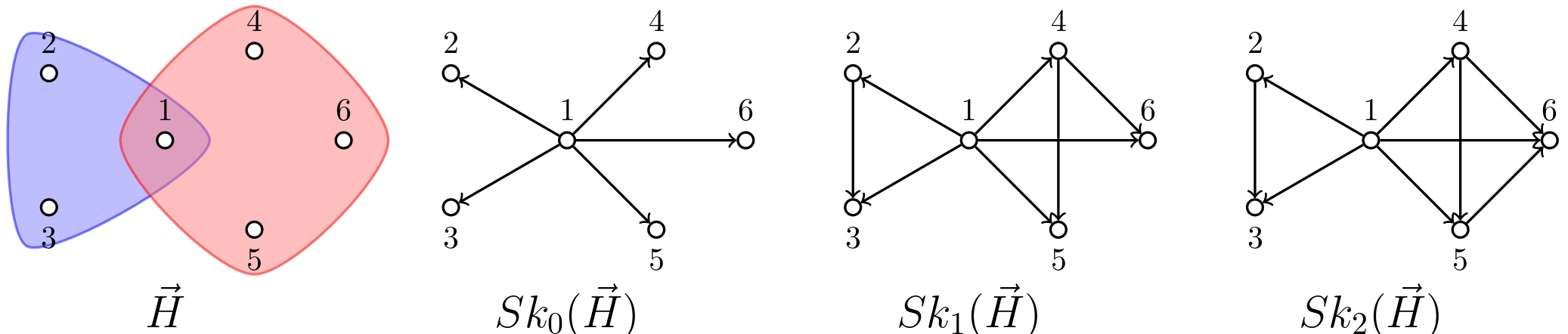}
	\caption{An example of a directed acyclic hypergraph following the ordering of the vertex numbering with all its different $l$-skeletons.}
	\label{fig:skeleton}
\end{figure}

\subsection{The degeneracy orientations}

A property of degenerate graphs is that one can obtain an ordering of the vertices such that orienting the graph using such ordering yields a directed acyclic graph with maximum out-degree equal to the degeneracy. Moreover, Matula and Beck showed that such an orientation can be constructed in linear time\cite{MaBe83}.

In \Sec{degeneracy} we show that this classical result can be extended to $l$-degenerate hypergraphs. First, we define the $l$-outdegree of a vertex in a DAH $\vec{G}$ as its outdegree in the $l$-skeleton of $\vec{G}$. We show that we can construct an orientation of any bounded $l$-degeneracy hypergraph $G$ that has bounded $l$-outdegree.

\begin{restatable}{lemma}{degeneracyordering} \label{lem:degeneracy_ordering}
   Let $\lparam \in \mathbb{Z}^+\cup \{\infty\}$ and let $G$ be a hypergraph with $\lparam$-degeneracy $\ldegen{l}(G)$ and rank $r(G)$. There exists an ordering of the vertices of $V(G)$ such that the directed acyclic hypergraph $\vec{G}^{(l)}$ obtained by applying this ordering to $G$ has a maximum $l$-outdegree $\Delta^+_l(\vec{G}^{(l)}) \leq \ldegen{\lparam}(G) \cdot r(G)$. Moreover, given $G$, one can construct $\vec{G}^{(l)}$ in $O(n+m(\log{m} \cdot r(G) + r(G)^2))$ time.
\end{restatable}

\subsubsection{\dagtree{} for hypergraphs}

Bressan introduced the concepts of \dagtree{} and \dagtw{} \cite{Br19,Br21} in the context of counting homomorphisms of degenerate graphs. The \dagtree{} gives a way of partitioning the pattern graph in order to compute the homomorphism counts efficiently from the counts of each partition, while the \dagtw{} is the maximum size of the bags of sources of the tree decomposition, and gives an upper bound on the complexity of counting the pattern. It was proved that we can only find linear time algorithms for counting patterns with a \dagtw{} of $1$.

We extend this notion to $l$-degenerate hypergraphs. We define the $l$-\dagtw{} ($\ldtw{l}$) of a directed acyclic hypergraph $\vec{H}$ as the \dagtw{} of its $l$-skeleton (see \Def{dtw} for more details).

In \Sec{bressan} we show how to use these decompositions to compute the homomorphisms of directed acyclic hypergraphs, extending the homomorphism counting result of Bressan \cite{Br19} to hypergraph homomorphism counting, for all notions of $l$-degeneracy and $l$-\dagtw. 

\begin{restatable}[]{theorem}{computing}
\label{thm:computing}
    Let $\lparam \in \Z^+ \cup \{\infty\}$ and $H, G$ be two bounded rank hypergraphs. There is an algorithm that computes $\Hom{G}{H}$ in time $\poly(\ldegen{l})n^{\ldtw{l}(H)} \log{m}$.
\end{restatable}

\subsubsection{Characterizing the hypergraphs with l-\dagtw{} equal to $1$}

To find near linear time countable instances of $H$, we analyze which instances of $H$ will have $\ldtw{\lparam} = 1$ in all its acyclical orientations. In \Sec{characterize} we prove that being $\setH{\lparam}$ ITS free is the exact condition for having  $l$-\dagtw{} equal to $1$. This, together with \Thm{computing} gives an upper bound in the complexity of counting homomorphisms of $\setH{\lparam}$ ITS free hypergraphs, giving the upper bound of \Thm{main}.

\begin{restatable}[]{theorem}{characterize}
\label{thm:characterize}  
   Let $\lparam \in \Z^+ \cup \{\infty\}$ and $H$ a hypergraph. $\dtw_l(H)=1$ if and only if $H$ is $\setH{l}$ ITS free.
\end{restatable}

\subsubsection{The hardness}

In \Sec{hardness}, we prove that our characterization is tight. The proof combines the color-coding techniques introduced in \cite{AlYuZw94}, with a construction that allows us to relate colorful homomorphism counting of patterns in $\setH{\lparam}$ with finding colorful simplices. This theorem, together with \Conj{simplex} gives the hardness result of \Thm{main}.

\begin{restatable}[]{theorem}{hardness}
    Let $\lparam \in \mathbb{Z}^+\cup\infty$ and let $H$ be a hypergraph that contains $H' \in \setH{l}$. If there exists an algorithm that computes $\Hom{G}{H}$ in time $f(\ldegen{l})O(n^{\gamma})$ for all $G$ and some $\gamma>1$, then there is an algorithm that, for some $k\geq 2$, detects if a regular arity $k+1$ hypergraph $G$ contains a $k$-simplex in time $\tilde{O}(m^{\gamma})$.
\label{thm:hardness}
\end{restatable}

\subsubsection{Computing subhypergraph counts}

Finally, in \Sec{sub} we extend the homomorphism result to subhypergraphs. The number of subhypergraphs of a pattern can be written as a linear combination of the homomorphism counts of each hypergraph in its quotient set (see \Lem{hombasis}). This directly gives the upper bound of \Thm{sub}.

In a recent paper, Bressan et al. showed that one can extend the hardness of hypergraph homomorphism counting to subhypergraph counting using the \emph{homomorphism basis}~\cite{BrBrDe+25}, analogous to the graph case \cite{CuDeMa17}. We show that one can also use the homomorphism basis to extend more fine-grained bounds in hypergraphs.
\newpage
\section{Preliminaries}

\subsection{Hypergraphs and homomorphisms} A \textbf{hypergraph} $G$ is a generalization of a graph defined by a vertex set $V(G)$ and a hyperedge set $E(G)$ where every hyperedge $e \in E(G)$ is a non-empty subset of the vertices in $V(G)$ \footnote{An alternative way of defining hypergraphs uses an incidence function $f_G$ which maps every edge $e \in E(G)$ to a subset of $V(G)$. In our context, both definitions are equivalent.}. We assume that the hypergraphs are simple, that is, they do not contain duplicated hyperedges.

The \textbf{arity} of a hyperedge $e$ is the number of vertices that is incident to, $|e|$. The rank $r(H)$ of a hypergraph $H$ is the maximum arity of its hyperedges. Graphs corresponds to the special case of hypergraphs of regular arity $2$ (where all hyperedges have arity exactly equal to $2$).
The \textbf{degree} $d_{H}(v)$ of a vertex $v$ in a hypergraph $H$ is the number of hyperedges that contain the vertex. If $H$ is clear from the context, we might write $d(v)$ instead.

A \textbf{homomorphism} from a hypergraph $H$ to a hypergraph $G$ is a map $\phi: V(H) \to V(G)$ such that for each $e \in E(H)$ we have $\{\phi(v):v \in e\} \in E(G)$, that is, for every hyperedge in $H$ the image of its vertices in $G$ corresponds exactly with one of its hyperedges\footnote{An alternative way to define homomorphisms in hypergraphs uses two maps, one for vertices and one for edges, however, because we assume our input graph to be simple, this is not necessary.}. Note that a hyperedge in $H$ can be mapped to an edge in $G$ of equal or lesser arity, but never to an edge of greater arity. Hence, we can ignore all edges of $G$ with arity greater than the rank of $H$. We also define \textbf{arity-preserving homomorphisms} as a map $\phi': V(H) \to V(G)$ such that $\forall e \in E(H), \exists e' \in E(G) \text{ s.t. }\{\phi(v):v \in e\} = e' \text{ and } |e| = |e'|$.

We use $\Phi(H,G)$ to denote the collection of homomorphisms from $H$ to $G$. We use $\Hom{G}{H}$ for the number of homomorphisms from $H$ to $G$, $|\Phi(H,G)|$. Note that in the case of hypergraphs with arity $2$ in all their hyperedges, this problem corresponds to the standard homomorphism counting problem in graphs. We use $\Homr{G}{H}$ for the number of arity-preserving homomorphisms from $H$ to $G$.

We use $H^c$ to denote the \textbf{clique completion} of $H$, that is, the graph obtained by replacing every hyperedge $e$ of $H$ by a clique of size $|e|$. $E(H^c) = \{e' \subseteq e | e \in E(H), |e'| = 2\}$.

Additionally, we will use the following notation regarding homomorphisms:
\begin{itemize}
	\item We use $V(\phi)$ to represent the \textbf{domain} of a homomorphism, that is, the vertices being mapped by it.
	\item We say that two homomorphisms $\phi , \phi'$ \textbf{agree} if for every vertex $v$ in $V(\phi) \cap V(\phi')$, $\phi(v) = \phi'(v)$.
	\item The \textbf{restriction} of $\phi$ to a set $V' \subseteq V(\phi)$, is the homomorphism $\phi'$ with $V(\phi')=V'$ and $\phi'(v) = \phi(v)$ for all $v \in V'$.
	\item Let $H'$ be a subhypergraph of $H$. The \textbf{extension} of a homomorphism $\phi:V(H')\to V(G)$ to $H$, $\ext{\phi}{H}{G}$, is the number of homomorphisms $\phi': V(H) \to V(G)$ that agree with $\phi$.
\end{itemize}

\subsection{Subhypergraphs}

In the standard notion of subhypergraphs, every hyperedge of the subhypergraph must be completely contained in the set of vertices of the subhypergraph. Formally, a hypergraph $H'=(V(H'),E(H'))$ is a \textbf{subhypergraph} of $H=(V(H),E(H))$ if $V(H')\subseteq V(H)$ and $E(H') \subseteq \{e \in E(H) | e \subseteq V(H')\}$. Given two hypergraphs, we will use $\Sub{G}{H}$ for the number of subhypergraphs of $G$ that are isomorphic to $H$.

The quotient set of a hypergraph is the set of all possible hypergraphs that can be obtained by merging partitions of its vertex set. 

\begin{definition} [\textbf{Quotient set} ($\cQ(H)$) \cite{BrBrDe+25}] \label{def:quotient}
	Let $H$ be a hypergraph and $\tau = \{V_1,...,V_l\}$ a partition of its vertices. For a set of vertices $X\subseteq V$, the quotient of $X$ under $\tau$ is $X/\tau = \{V_i \in \tau : |V_i \cap X|>0\}$. For a hypergraph $H$, the quotient of $H$ under $\tau$ is the hypergraph $H/\tau = (V(H)/\tau, \{ e/\tau : e \in E(H)\})$.
	$\cQ(H)$ is the set of all the quotients of $H$.
\end{definition}

Note that the image of every homomorphism of $H$ in a graph $G$ corresponds to a subhypergraph of some pattern in the quotient set of $H$. \cite{BrBrDe+25} showed that like in the general graph case, this can be used to express the number of subhypergraphs of $H$ as a linear combination of the number of homomorphisms of all the patterns in its quotient set.

\begin{lemma} \cite{BrBrDe+25} \label{lem:hombasis}
	For all hypergraphs $H$ there is a function $\gamma : \cQ(H) \to \mathbb{Q}$ such that, for all hypergraph $G$ we have:
	\[
	\Sub{G}{H} = \sum_{H' \in \cQ(H)} \gamma(H') \Hom{G}{H'}
	\]
	Moreover, for all $H' \in \cQ(H)$, $\gamma(H') \neq 0$.
\end{lemma}

%

\subsection{Induced Subhypergraphs and Induced Trimmed Subhypergraphs}

\begin{definition}[Induced subhypergraph]
	For any subset $S \subseteq V(G)$, the induced subhypergraph of $G$ on $S$ is the hypergraph $G ( S ) = (S, E( S ))$ with $E( S ) =  \{ e \in E(H) \mid e\subseteq S \}$.
\end{definition}

We can obtain a more \emph{relaxed} definition of an induced subhypergraph by including hyperedges that are only partially included in the vertex set of the subhypergraph. This was called a shrinkage in \cite{ScSt21} and a trimmed subhypergraph in \cite{BrBrDe+25}. We will refer to it as the latter in this paper.

\begin{definition} [Induced Trimmed Subhypergraph]
	For any subset $S \subseteq V(G)$, the induced trimmed subhypergraph of $G$ on $S$ is the hypergraph $G \langle V \rangle = (S, E\langle S \rangle)$ with $E\langle S \rangle =  \{ e \cap S \mid e \in E(H),\ e \cap S \not= \emptyset \}$.
\end{definition}

Note that, if $H$ is a graph, then all induced trimmed subhypergraphs are also induced subhypergraphs. Let $\cH$ be a set of hypergraphs, we say that the hypergraph $H$ is $\mathbf{\cH}$\textbf{-ITS free} (Induced trimmed subhypergraph free) if no induced trimmed subhypergraph of $H$ is isomorphic to $H'$.

We can also define intermediary notions of induced subhypergraphs where we allow edges that contain only a limited number of vertices that are not part of the vertex set of the subhypergraph. We call this an induced $l$-trimmed subhypergraph.

\begin{definition} [Induced $\lparam$-trimmed subhypergraph] 
	Let $\lparam \in \Z^+ \cup \{\infty\}$. For any subset $S \subseteq V(G)$,
	the induced $\lparam$-trimmed subhypergraph of $G$ on $S$ is the hypergraph $G \langle S \rangle_\lparam = (S, E \langle S \rangle_\lparam)$ with 
	$E \langle S \rangle_\lparam = \{e \cap S \ | \ e \in E(G), |e \setminus S| \leq \lparam, e \cap S \not = \emptyset\}$.
\end{definition}

Note that an induced subhypergraph is equivalent to an induced $0$-trimmed subhypergraph, and a trimmed subhypergraph corresponds to a $\infty$-trimmed subhypergraph. Additionally, note that for a fixed subset $S \subseteq V$, $E \langle S \rangle_\lparam \subseteq E \langle S \rangle_{(\lparam+1)}$.  

\subsection{Directed Acyclic Hypergraphs}

Like in the standard graph homomorphism counting problem, we will be exploiting specific orientations of the input graph to compute the number of homomorphisms. Hence, we need to define directed hypergraphs. In the most common definition of a \textbf{directed hypergraph} \cite{AuLa17}, each hyperedge is defined by two sets of vertices: the tail and the head; this directly generalizes directed graphs where, for every directed edge, one vertex is the tail and another the head. 

For our purposes, we will think of directed hypergraphs as given by an ordering of its vertices. Each directed hyperedge will hence be given by an ordered list of vertices, which matches the ordering of the hypergraph. Directed hypergraphs defined this way are equivalent to directed acyclic graphs, and we will call them Directed Acyclic Hypergraphs. A \textbf{Directed Acyclic Hypergraph (DAH)} $\vec{G}$, is a hypergraph $H$ with an ordering of its vertices $\ord: V(G) \to [V(G)]$, every hyperedge $e \in E(G)$ becomes an ordered list in $E(\vec{G})$, where the order of each hyperedge follows the ordering of $G$.

We use $\Sigma(H)$ to denote the set of acyclical orientations of $H$, up to isomorphism.

For a directed acyclic hypergraph, a \textbf{source} is a vertex that is first in the ordering of all hyperedges to which it belongs. We use $S(\vec{H})$ to denote the set of sources of the directed acyclic hypergraph $\vec{H}$.

For a directed acyclic hypergraph $\vec{H}$ we define its \textbf{skeleton} $\Skeleton{0}{\vec{H}}$ as the directed acyclic graph obtained by replacing every directed hyperedge by an out-star centered at the first vertex of the edge, that is, a set of directed edges connecting the first vertex in the ordering of the hyperedge to every other vertex in it. 

Similarly, we can define the $\mathbf{l}$\textbf{-skeleton} $\Skeleton{l}{\vec{H}}$ of $\vec{H}$ as the graph obtained by replacing every directed hyperedge $e$ by a set of directed edges connecting each of the first $l+1$ vertices in $e$ to all vertices that appear after it in the internal ordering of $e$. See \Fig{skeleton} for an example.

Note that the $\infty$-skeleton of a hypergraph $\vec{H}$ corresponds to the acyclic orientation of the clique completion of $\vec{H}$ that matches the vertex ordering of $\vec{H}$. Note that for all $l$, the set of sources of a hypergraph and its $l$-skeletons are the same. $S(\vec{H})=S(\Skeleton{l}{\vec{H}})$.

Given two directed acyclic hypergraphs $\vec{H},\vec{G}$, a \textbf{DAH homomorphism} is a mapping $\phi: V(\vec{H}) \to V(\vec{G})$, such that for every directed hyperedge $e \in E(\vec{H})$ we have a $\phi(e) \in E(\vec{G})$. Note that because the directed hyperedges are ordered lists, they will necessarily preserve arity. Moreover, note that for $e = (v_1, v_2,..., v_{|e|})$, we will have $\phi(e) = (\phi(v_1), \phi(v_2),..., \phi(v_{|e|})$, and the internal hyperedge ordering will also be preserved.

\subsection{Reachability and outdegree}

The concept of reachability in graphs is central to the homomorphism counting techniques based on \dagtw{} (which we will define next). In a directed graph, we say that a vertex $u$ \textbf{reaches} $v$ if there is a direct path connecting $u$ to $v$. However, in a directed hypergraph, we can think of reachability in different ways.

Using the standard definition of directed hypergraphs, reaching the tail set of a hyperedge allows one to reach the vertices of the head. The most naive way to think about reachability in a directed acyclic hypergraph $\vec{H}$ would be to have the first vertex in the internal ordering of an edge reach all the other vertices in the edge; this would match the reachability of the $0$-skeleton of $\vec{H}$. It is also equivalent to the standard notion of reachability in hypergraphs when setting the first vertex in each hyperedge as tail and the rest as head. 

Furthermore, we can consider alternative notions of reachability by instead allowing the first $l+1$ vertices in a hyperedge reach the vertices appearing after it. Note that this would be equivalent to the reachability of the $l$-skeleton of the directed acyclic hypergraph.

We say that a vertex $u$ $\lparam$-\textbf{reaches} $v$ in a DAH $\vec{H}$, if $u$ reaches $v$ in the $\lparam$-skeleton of $\vec{H}$.

In a directed acyclic graph $\vec{G}$, we will use $\Reach_{\vec{G}}(s)$ for the set of vertices reachable by $s$ in $\vec{G}$, and $\Reach_{\vec{G}}(S)$ for the union of $\Reach_{\vec{G}}(s)$ for all $s \in S$. Moreover, we will use $\vec{G}[s]$ and $\vec{G}[S]$ for the subgraph of $G$ induced by $\Reach_{\vec{G}}(s)$ and $\Reach_{\vec{G}}(S)$ respectively. 

Similarly for a DAH $\vec{H}$, we will $\Reach_{\vec{H}, \lparam}(s)$ (or $\Reach_{\vec{H}, \lparam}(s)$ if clear from the context) for the set of vertices $l$-reachable by $s$ in $\vec{H}$. Note that $\Reach_{\vec{H}, \lparam}(s) = \Reach_{\Skeleton{\lparam}{\vec{H}}}(s)$. We use  $\Reach_{\vec{H}, \lparam}(S)$ for the union of $\Reach_{\vec{H}, \lparam}(s)$ for all $s \in S$. And we will use $\vec{H}_{\lparam}[s]$ and $\vec{H}_{\lparam}[S]$ for the ($0$-trimmed) subhypergraphs of $G$ induced by the vertices in $\Reach_{\vec{H}, \lparam}(s)$ and $\Reach_{\vec{H}, \lparam}(S)$ respectively.

Similarly to how we did for reachability, we define the $\mathbf{\lparam}$\textbf{-outdegree} of a vertex $v$ in a DAH $\vec{H}$, $d^+_{\vec{H},\lparam}(v)$ as the outdegree in the $l$-skeleton. We will use $\Delta^+(\vec{H})$ to denote the maximum $l$-outdegree of $\vec{H}$.

\subsection{The DAG-tree decomposition}

Bressan \cite{Br19, Br21} introduced the concept of \dagtree{} of a directed acyclic graph to find efficient algorithms for homomorphism counting in degenerate graphs.

\begin{definition}[\dagtree{} \cite{Br19}]
    For a given DAG $\vec{H}$, a \dagtree{} of $\vec{H}$ is a rooted tree $T=(\cB,\cE)$ such that:
    \begin{itemize}
        \item Each node of the tree $B \in \cB$ is a subset of the sources of $\vec{H}$, $B \subseteq S(\vec{H})$.
        \item Every source $s \in S(\vec{H})$ appears in at least one node $B \in \cB$, $\bigcup_{B \in \cB} B = S(\vec{H})$.
        \item $\forall B, B_1, B_2 \in \cB$. If $B$ is in the (unique) path between $B_1$ and $B_2$ in $T$, then $\Reach_{\vec{H}}(B_1) \cap \Reach_{\vec{H}}(B_2) \subseteq \Reach_{\vec{H}}(B)$. 
    \end{itemize}
\end{definition}

The \dagtw{} $\dtw(T)$ of a \dagtree{} $T$ is the maximum size among all the nodes in $\cB$. The \dagtw{} $\dtw(\vec{H}$ of a DAG $\vec{H}$ is the minimum $\tau(T)$ of all valid \dagtree{} $T$ of $\vec{H}$. There exists a relation between \dagtw{} and $\alpha$-acyclicity.

\begin{definition}[$\alpha$-acyclicity]
    A hypergraph $H$ is $\alpha$-acyclic if there exists a tree $T$ whose vertices are the hyperedges of $F$, such that, for all $e_1,e_2,e_3 \in E(H)$ such that $e_2$ is in the unique path of $T$ between $e_1$ and $e_3$, we have $e_1 \cap e_3 \subseteq e_2$.
\end{definition}

The \textbf{reachability hypergraph} of a dag $\vec{H}$, is the hypergraph $H'$ on vertex set $V(H') = \vec{H}$ that contains, for every source $s \in S(\vec{H})$ a hyperedge $e_s = \Reach_{\vec{H}}(s)$. A dag $\vec{H}$ having $\dtw(\vec{H})=1$ is equivalent to its reachability hypergraph being $\alpha$-acyclic.

\begin{lemma}\label{lem:alpha-acyclic}
	A directed acyclic $\vec{H}$ graph has $\dtw(\vec{H})$ iff the reachability hypergraph of $\vec{H}$ is $\alpha$-acyclic.
\end{lemma}

We can generalize the concept of \dagtree{} to directed acyclic hypergraphs. Note that the definition of \dagtree{} only depends on the reachability function of the graph. We could define multiple notions of \dagtree{s} depending on which $\lparam$-reachability of the DAH $\vec{H}$ we are using, however this turns out to be equivalent to computing \dagtree(s) for each possible $\lparam$-skeleton of $\vec{H}$. Therefore, we will instead define notions of \dagtw{} using the \dagtree of the $\lparam$-skeletons.

\begin{definition} [$\ldtw{\lparam}$] \label{def:dtw}
	Let $\lparam \in \mathbb{Z}^+\cup \{\infty\}$. The $\lparam$-\dagtw{} of a directed acyclic hypergraph $\vec{H}$ is the \dagtw{} of its $\lparam$-skeleton.
	\[
		 \ldtw{\lparam}(\vec{H}) = \dtw(\Skeleton{\lparam}{\vec{H}})
	\]
	 The $\lparam$-\dagtw{} of an undirected hypergraph $H$ is the maximum $\lparam$-\dagtw{} across all DAH $\vec{H} \in \Sigma(H)$.
\end{definition}

The \textbf{down-closure} of a bag $B$ in a \dagtree{} $\cT$, $\down(B)$, is the union of all the bags in the sub-tree of $\cT$ rooted at $B$. We will use $\vec{H}[\down(B)]_l$ for the subhypergraph of $H$ induced by all the vertices in $\Reach_{\vec{H},l}(\down(B))$.

\subsection{Degeneracy}

The degeneracy of a graph is a measurement of its sparsity. A graph $G$ is $\kappa$-degenerate if all its subgraphs have minimum degree of at most $\kappa$, while the degeneracy of $G$ is the maximum $\kappa$ such that $G$ is $\kappa$-degenerate. A similar notion exists for hypergraphs \cite{KoRo06}. 

\begin{definition} [Degeneracy]
    We say that a hypergraph $H$ is $\kappa$-degenerate if every induced subhypergraph of $H$ has a minimum degree of at most $\kappa$. The degeneracy $\degen(H)$ of a hypergraph $H$ is the minimum value of $\kappa$ such that $H$ is $\kappa$-degenerate.
\end{definition}

We can define additional notions of degeneracy by using induced $\lparam$-subhypergraphs instead of induced subhypergraphs.

\begin{definition} [$\lparam$-degeneracy]
    We say that a hypergraph is $\kappa$-$l$-degenerate if every induced $l$-trimmed subhypergraph of $G$ has a minimum degree of at most $\kappa$. The $l$-degeneracy $\ldegen{\lparam}(G)$ of a hypergraph $H$ is the minimum value of $\kappa$ such that $H$ is $\kappa$-degenerate.
\end{definition}

We have $\degen(G) = \ldegen{0}(G)$. Also, if $G$ has bounded rank, the $\ldegen{\lparam}(G)$ will be bounded by $O(\degen(G^c))$. Furthermore, note that the $\lparam$-degeneracy increases with $\lparam$, since the induced $\lparam$-trimmed subhypergraphs are subhypergraphs of the induced ($\lparam+1$)-trimmed subhypergraphs. Thus, for all $\lparam$, $\ldegen{\lparam}(G) \leq \ldegen{\lparam}(G)$.

\section{The degeneracy orientation of hypergraphs} \label{sec:degeneracy}

A property of degenerate graphs is that one can obtain an ordering of the vertices such that orienting the graph using such ordering yields a directed acyclic graph with maximum outdegree equal to the degeneracy. Moreover, such an orientation can be constructed in linear time \cite{MaBe83}.

We can show that this classical result can be extended to all the notions of $\lparam$-degeneracy that we have defined, giving orientations with maximum $l$-outdegree bounded by the $l$-degeneracy. We will use $\vec{G}^{(\lparam)}$ to denote the $\lparam$-degeneracy orientation of $G$.

\degeneracyordering*

\begin{algorithm} 
	\caption{$Compute\_ordering(G,l)$}
	\begin{algorithmic}[1] 
		\State Set $G' = G$, $V' = V(G)$.
		\For{$i \in [n]$}
			\State Set $v_i$ as a vertex of $G'$ with minimum degree
			\State Set $V' = V \setminus \{v_i\}$
			\State Set $G' = G\langle V'\rangle_l$
		\EndFor
		\State Return $v_0,...,v_{n-1}$
	\end{algorithmic}
\end{algorithm}\label{alg:ordering}

\begin{proof}
    Fix $\lparam$. We can show that algorithm $Compute\_ordering$ (\Alg{ordering}), gives the desired ordering. Moreover, we show how to implement it efficiently. The algorithm follows a similar structure to the original ordering algorithm from \cite{MaBe83}.
    
    \begin{itemize}
    \item Step $3$: We use the degree structure from \cite{MaBe83}: for each possible degree, we create a doubly linked list of vertices of degree $i$ for each $i$, which we initialize putting each vertex in the appropriate list based on their degree in $G$. This structure allows to find a vertex $v_i$ with minimum degree at any point in time $O(d(v_i))$, where $d(v_i)$ is the degree at the time of deletion. It also allows us to update the degree of a vertex in constant time. We can initialize the structure in $O(n+m)$ time.
    
 	\item Step $4$: Updating the vertex set can be done in constant time, by maintaining an array of active vertices.
 	
 	\item Step $5$: We need to update $G'$ by removing $v_i$ from each hyperedge. To do this, we keep a dictionary that maps a tuple of vertices to each active hyperedge, this represents the trimmed portion of each hyperedge left in the induced $\lparam$-trimmed subhypergraph. Additionally, for each hyperedge, we keep a counter of how many of its vertices have been removed.
 	
 	We iterate over all the hyperedges $e \in E(G)$ that contain $v_i$ (at most $d_G(v_i)$) and increase the counter for all of them.
 	If any hyperedge has lost more than $l$ vertices, then we remove it and update the degrees of the vertices on it (at most $r(e)$). Otherwise we compute the tuple corresponding to $e \cap V'$ \footnote{We can use an arbitrary total order on the vertices so the tuple is unique for each possible subset.}, this can be done in $O(r(e))$ time. If that tuple is not in the dictionary, we add it. Otherwise, it means that the hyperedge is already present in $G'$, we use the dictionary to retrieve the conflicting hyperedge, and keep the one with lower counter of deleted vertices, deleting (and updating the degrees) of the other, which will take $O(r(e))$ time.
 	
 	The entire process of step $5$ will take $O(r(e) + \log{m})$ per hyperedge, where $\log{m}$ is the dictionary cost. 
    
    \end{itemize}
    Therefore, steps $3$-$5$ will take a total of $d_G(O(v_i)(r(G) + \log{m}))$ time. The sum of degrees in a hypergraph is bounded by the rank times the number of hyperedges, therefore, the total runtime of the algorithm will be bounded by $O(n+m(\log{m} \cdot r(G) + r(G)^2))$.
    
	Now, let $\vec{G}^{(\lparam)}$ be the directed acyclic hypergraph obtained from $G$ and using the deletion ordering $v_0,...,v_{n-1}$. We will have the following.
    \begin{equation}
        d^+_{\vec{G}^{(\lparam)},\lparam}(v_i) \leq d_{G\langle \{v_i,...,v_n\}\rangle_{\lparam}}(v_i)\cdot r(G) \leq \ldegen{\lparam}(G)  \cdot r(G)
    \end{equation}
	The first inequality holds from the fact that the $\lparam$-outdegree of the vertex $v_i$ will be bounded by its degree in the induced $\lparam$-trimmed subhypergraph at the time of deletion times the rank of the hypergraph. The second inequality holds because $v_i$ is a vertex with minimum degree in $G\langle \{v_i,...,v_n\}\rangle_{\lparam}$, and by definition the minimum degree of any induced $\lparam$-trimmed subhypergraph is bounded by $\ldegen{\lparam}(G)$. Finally, we will have:
	\begin{equation}
		\Delta^+(\vec{G}^{(\lparam)}) = \max_{v \in V(G)} d^+_{\vec{G}^{(\lparam)},\lparam}(v_i) \leq \ldegen{\lparam}(G) \cdot r(G)
	\end{equation}
\end{proof}

\section{Counting hypergraph homomorphisms in degenerate hypergraphs} \label{sec:bressan}

\subsection{Reducing to directed homomorphisms}

Recall from the definition of homomorphism in hypergraphs that we are allowed to map hyperedges to other hyperedges of smaller arity. If both the pattern and the input have the exact same arity for all their hyperedges, then this will not be an issue, but when considering hyperedges of any arity, this can become problematic.

\begin{definition} [$\contract(H)$]
	For a hypergraph $H$, the contract set $\contract(H) \subseteq \cQ(H)$, is the set of all the quotient hypergraphs $H/\tau$ generated by a partition $\tau = \{ V_1,V_2,...V_l \}$ such that every set $V_i$ forms a connected component in $H$.
\end{definition}

Alternatively, one can generate $\contract(H)$ by recursively merging two vertices in the same hyperedge. Note that every homomorphism from $H$ to an input hypergraph $G$ will correspond to one homomorphism in $\contract(H)$; hence, we can obtain $\Hom{G}{H}$ by computing the number of arity-preserving homomorphisms for each pattern in $\contract(H)$.

\begin{lemma} \label{lem:contracted}
	For all graph $H$:
	\[
    \Hom{G}{H} = \sum_{H' \in \contract(H)}  \beta(H') \Homr{G}{H'}
    \]
    Where $\beta(H')>0$ for all $H'$ and can be computed in $O(1)$ time.
\end{lemma}
\begin{proof}
	Let $\phi$ be a homomorphism from $H$ to $G$, for every vertex $v$ in the image of $\phi$, let $V_v = \{u \in V(H) : \phi(u) = v \}$, divide each set $V_v$ into connected components of $H \langle V_v \rangle$ and set $\tau = \{V_1,...,V_l\}$ be the set of all connected components. Note that $\tau$ is a valid partition of $V(H)$, and $H'=V/\tau \in \contract(H)$. Let $\phi': \tau \to V(G)$ by assigning $v$ to every set obtained from $V_v$. 
	
	We can show that $\phi'$ must be an arity preserving homomorphism from  $H/\tau$ to $G$. Consider any hyperedge $e$ in $H/\tau$, every vertex $V_i,V_j$ in $e$ must have a different image $\phi'(V_i) \neq \phi'(V_j)$, as otherwise they would be part of the same connected component of $V_v$ and would not have been split into different sets. Moreover, $e$ corresponds to at least a hyperedge $e'$ in $H$, and $\phi'(e) = \phi(e') \in E(G)$.
	
	Let $\alpha(H')$ be the number of valid partitions $\tau$ of $H$ such that $H/\tau$ is isomorphic to $H'$. For each arity preserving homomorphism of $H'$ and each such $\tau$ we can obtain a homomorphism $\phi: H \to G$ by setting $\phi(u) = \phi'(V_i \ni u)$ for each $u$, moreover, for each automorphism (map of $H'$ to itself that yields an isomorphic hypergraph) of $H'$ we will have a homomorphism $\phi'$ that generates the same maps $\phi$. Therefore, we can set $\beta(H') = \alpha(H')/Aut(H')$ to obtain the expression of the lemma.
\end{proof}

We can show that the patterns in $\contract(H)$ cannot have a greater $\lparam$-\dagtw{} than $H$. Note that this is not true in general for the entire set of quotients of $H$, $\cQ(H)$.

\begin{lemma} \label{lem:dtw_bound}
    Let $\lparam \in \mathbb{Z}^+\cup \{\infty\}$ and $H$ be a hypergraph. For every hypergraph $H' \in \contract(H)$ $\ldtw{\lparam}(H') \leq \ldtw{\lparam}(H)$.
\end{lemma}
\begin{proof}
    Fix $\lparam$. Let $H$ be a hypergraph with $\tau =\ldtw{l}(H)$. We show that all hypergraphs $H'$ in $\contract(H)$ have $\ldtw{l}(H') \leq \tau$, we do induction on the number of vertices of $H'$. The base case, $|V(H')| = |V(H)|$, is trivial as $H' = H$, and therefore $\ldtw{l}(H') = \tau$.

    For the inductive step, assume that for all hypergraphs $H'$ in $\contract(H)$ with $k$ vertices we have $\ldtw{l}(H') \leq \tau$, we show that any hypergraph $H'$ in $\contract(H)$ with $k-1$ vertices will also have $\ldtw{l}(H') \leq \tau$.

    Fix a hypergraph $H'$ in $\contract(H)$ with $|V(H')| = k-1$. We show that any orientation $\vec{H'}$ of $H'$ has $\ldtw{l}(\vec{H'})\leq \tau$.

    Let $H'' \in \contract(H)$ with $|V(H'')| = k$ be a hypergraph such that merging the vertices $u,v \in V(H'')$ results in $H'$. Such a hypergraph must exist; otherwise $V(H')$ would not be in $\contract(H)$. Moreover, from the induction hypothesis, we have $\ldtw{l}(V(H'')) \leq \tau$. We call $uv$ to the vertex in $H'$ resulting from merging $u$ and $v$.

    Let $\ord': V(H') \to [k-1]$ be the ordering of the vertices of $H'$ according to $\vec{H'}$. We define $\ord'': V(H'')\to[k]$ as follows:
    \begin{equation}
        \ord''(w)= 
    \begin{cases}
    \ord'(w) & \text{if } \ord'(w) < \ord'(uv) \\
    \ord'(uv) & \text{if } w=u\\
    \ord'(uv)+1 & \text{if } w=v\\
    \ord'(w)+1 & \text{if } \ord'(w) > \ord'(uv)
    \end{cases}
    \end{equation} 
    
    Let $\vec{H}''$ be the directed acyclic hypergraph obtained from $H''$ using the ordering $\ord''$. Consider the $l$-skeleton of $\vec{H}'$ and $\vec{H}''$, $Sk'=\Skeleton{l}{\vec{H'}}$ and $Sk''=\Skeleton{l}{\vec{H''}}$. Let $\cT''$ be a \dagtree{} of $Sk''$ with $\dtw(\cT'') \leq \tau$, such \dagtree{} must exist from the definition of $\ldtw{l}$. We show that $Sk'$ admits a \dagtree{} $\cT'$ with $\dtw(\cT') \leq \tau$. The direct arc $(u,v)$ might or might not be present in $Sk''$. We consider each case separately.
    
    \paragraph{$\mathbf{(u,v) \not\in E(Sk'')}$} In this case, $Sk'$ will be the result of merging $u$ and $v$ in $Sk''$. Note that $u$ and $v$ share at least one hyperedge $e \in E(\vec{H}'')$. Let $s$ be the source of $e$, the direct arcs $(s,u)$ and $(s,v)$ must both be present in $Sk''$. We can show that $\cT''$ will also be a valid \dagtree{} for $Sk'$, and thus $\dtw(Sk)\leq \tau$. It suffices to show that for every vertex $w$ of $Sk'$, the sub-tree of $\cT''$ induced by the sources that reach $w$ is connected:
    	\begin{itemize}
    		\item $w=uv$. $w$ will be reached by both the bags that were reaching $u$ and the bags that were reaching $v$. Note that the bags of sources reaching $u$ will form a connected sub-tree, similarly with the bags of sources reaching $v$. Moreover, both sub-trees intersect in the bags of sources that can reach $s$ (or that contain $s$ if $s$ is a source of the skeleton), which means that they must form a single connected sub-tree.
    		\item $w \in \Reach_{Sk''}(u)$. $w$ will be reached by the same bags as in $Sk''$ plus all the bags that were reaching $v$. Both will form connected sub-trees in $\cT''$, and again they must intersect in at least the bags that reach $s$, which means that they will form a connected sub-tree.
    		\item $w \in \Reach_{Sk''}(v)$. Analogous to the previous case.
    		\item $w \not\in (\Reach_{Sk''}(u) \cup \Reach_{Sk''}(c))$. In this case the set of sources reaching $w$ in $Sk'$ is the same as in $Sk''$, and therefore they will form a connected sub-tree.
    	\end{itemize}
    	
    \paragraph{$\mathbf{(u,v) \in E(Sk'')}$} In this case $Sk'$ is the result of merging the direct arc $(u,v)$. Moreover, it is possible that the merge might create new direct edges in $Sk'$. This can happen if the vertices $u,v$ are part of a hyperedge $e$ in $\vec{H''}$ and $v$ is one of the first $l+1$ vertices in the internal ordering of $e$ and $|e|>l+2$. In that case, $\vec{H'}$ will contain a hyperedge $e'$ where a vertex $u'$ that was in position $l+2$ in $e$ is now in the position $l+1$ in the internal ordering of $e'$, this means that for all $v'$ appearing after $u'$ in $e'$ we will have a direct arc $(u',v') \in Sk'$ that might not have been present in $Sk''$.
    
	We show that both merging the direct edge and adding a direct edge between two vertices on the same hyperedge do not increase $\dtw$.

    \begin{itemize}
        \item \textbf{Merging a direct arc: } We need to consider three sub-cases:
        \begin{enumerate}
            \item $u$ is a source in $Sk''$ and $uv$ is a source in $Sk'$. Replacing $u$ for $uv$ in every bag of $\cT''$ will result in a valid \dagtree{} $\cT'$ of $Sk'$, as $\Reach_{Sk'}(uv) \setminus \{uv\} = \Reach_{Sk''}(u) \setminus (u)$, and every other vertex will have the same reachability in both dags. No bag will increase in size and therefore $\dtw(\cT') \leq \dtw(\cT'') \leq \tau$.
            
            \item $u$ is a source in $Sk''$ but $uv$ is not a source in $Sk'$. This implies that $v$ had some incoming edges in $Sk''$, all the sources other than $u$ that reached $v$ in $Sk''$ now also reach $u$ and all the vertices in $\Reach_{Sk''}(u)$. Note that in $\cT''$ all bags containing $u$ must form a connected sub-tree, similarly all bags containing a source (including $u$) that reaches $v$ will form a connected sub-tree. Let $s$ be a source in $Sk''$ that reaches $v$, and either appears in a bag of $\cT''$ together with $u$, or is in a bag adjacent to a bag containing $u$. 
            
            We can form a \dagtree{} $\cT'$ of $Sk'$ by replacing $u$ by $s$ in every bag. No bag will increase in size and therefore $\dtw(\cT') \leq \dtw(\cT'') \leq \tau$. We can verify that this will be a valid \dagtree{}, let $w$ be a vertex in $V(Sk')$, we can show that the sub-tree of $\cT'$ that contains a source reaching $w$ is connected.
            \begin{itemize}
                \item If $w = uv$, then the bags that reach $uv$ in $Sk'$ are the same bags that reached $v$ in $Sk''$, which must form a connected sub-tree.
                \item If $w \in \Reach_{Sk''}(u)$, then $w$ is reachable by every bag that reached $v$ in $Sk''$ which forms a connected sub-tree, and every bag that reaches $w$ in $Sk''$, but those bags must form a connected sub-tree with the bags of $\cT''$, both sub-trees overlap in the bags that were containing $u$ (now $s$ instead), and therefore they will all form a bigger connected sub-tree.
                \item If $w \not\in \Reach_{Sk''}(u)$, then all bags that reached $w$ in $\cT''$ also reach $w$ in $\cT'$, if $s$ reaches $w$ then there might be some additional bags reaching $w$, but they will still form a connected sub-tree.
            \end{itemize}
            
            \item $u$ is not a source in $Sk''$. In this case we can show that $\cT''$ will be a valid \dagtree{} for $Sk'$. Let $w$ be a vertex in $V(Sk')$:
            \begin{itemize}
                \item If $w = uv$, then the bags that reach $w$ in $Sk'$ are the exact same bags that reached $v$ in $Sk''$, and therefore form a connected sub-tree.
                \item If $w \in \Reach_{Sk''}(u)$, note that all the sources reaching $v$ in $Sk''$ now will reach $w$ in $Sk'$. These sources form a connected sub-tree. All the sources that reached $w$ in $Sk''$ will also reach $w$ in $Sk''$ and therefore also form a connected sub-tree. Both sub-trees intersect in the bags containing sources that reach $u$ in $Sk''$, thus combining them yields a larger connected sub-tree.
                \item If $w \not\in \Reach_{Sk''}(u)$, then the sources that reach $w$ in $Sk'$  are the same sources that reached $w$ in $Sk''$, and therefore the bags containing those sources will form a connected sub-tree.
            \end{itemize}
        \end{enumerate}
        \item \textbf{Adding a direct edge: } We show that adding a direct edge between two vertices that are part of the same hyperedge does not increase $\dtw$. Let $\vec{F}$ and $\vec{F'}$ be two dags with the same vertex set and $E(\vec{F'}) = E(\vec{F}) + \{(u,v)\}$ for some vertices $u,v$ such that $\exists s\in V(\vec{F})$ with $(s,u),(s,v) \in E(\vec{F})$. Let $\cT$ be a \dagtree of $\vec{F}$, we can show that it will also be a valid \dagtree of $\vec{F'}$.
        
        Let $w$ be a vertex of $\vec{F'}$, we show that the bags reaching $w$ in $\vec{F'}$ must form a connected sub-tree:
        \begin{itemize}
        	\item If $w \in \Reach_{\vec{F}}(v)$, then $w$ will be reached in $\vec{F'}$ by the bags that reached $v$ and $u$. Both sets of bags form connected sub-trees in $\cT$ and intersect on the bags that reach $s$, therefore, they will all form a connected sub-tree.
        	\item If $w \not\in \Reach_{\vec{F}}(v)$, then the bags reaching $w$ are the same in $\vec{F}$ and $\vec{F'}$, therefore forming a connected sub-tree in $\cT$.
        \end{itemize}
    
   		Therefore, no matter how many direct edges $Sk'$ adds with respect to $Sk''$ the \dagtree{} will still be valid and we will have that $\dtw(Sk') \leq \dtw(Sk'')$. 
    \end{itemize}
\end{proof}

The next step is to exploit the degeneracy orientations of the input hypergraph. We will compute the degeneracy order of $G$ and obtain a directed acyclic hypergraph $\vec{G}^{(\lparam)}$ with $\Delta_{\lparam}^+(\vec{G}^{(\lparam)}) = O(\ldegen{\lparam}(G))$. Note that the image of each homomorphism of $H$ in $G$ appears now in a specific acyclic orientation, and corresponds to one of the possible directed acyclic hypergraphs that can be obtained from $H$. Hence, analogously to the graph case, we can compute the number of undirected hypergraph homomorphisms as the sum of the directed homomorphism for each orientation of $H$.

\begin{equation} \label{eq:undirect_to_direct}
    \Homr{G}{H} = \sum_{\vec{H} \in \Sigma(H)}\Hom{\vec{G}}{\vec{H}}
\end{equation}

\subsection{Counting directed hypergraph homomorphisms}

In this section, we show how to generalize Bressan's algorithm \cite{Br19,Br21} to our different notions of degeneracy and \dagtw{} to compute the number of DAH homomorphisms.

Let $\vec{H}$ and $\vec{G}$ be two directed acyclic hypergraphs. We could think of using the skeletons of $\vec{H}$ and $\vec{G}$ to compute the number of homomorphisms from $\vec{H}$ to $\vec{G}$. However, while for every valid homomorphism $\phi: V(\vec{H}) \to V(\vec{G})$ from $\vec{H}$ to $\vec{G}$ is also a valid homomorphism from $\Skeleton{l}{\vec{H}}$ to $\Skeleton{l}{\vec{G}}$, the opposite is not true. We could think of obtaining the homomorphisms from the skeletons and then just filtering out the ones that are not valid. This again does not directly work, as we would require to \emph{enumerate} the homomorphisms, rather than just counting, which can be potentially harder.

Alternatively, one could think of just computing the homomorphisms of the induced $l$-trimmed subhypergraph of each of the partitions of the \dagtree{}, but those homomorphisms will not match the homomorphisms of the entire pattern. Alternatively, one could try th induced ($0$-trimmed) subhypergraphs, but while the skeletons of the partitions and the $l$-trimmed subhypergraphs will always be connected, the induced subhypergraphs might not, and therefore we may not be able to compute the homomorphism list as efficiently.

Therefore we will need to do a hybrid approach, we will enumerate all the homomorphisms of the skeletons for each partition of the \dagtree{}, we will then filter that list by checking that every hyperedge of the corresponding induced subhypergraph is present. The result of this is a list of homomorphisms of the induced subhypergraphs for each of the partitions; this list does not include all such homomorphisms, as, as we said, the induced subhypergraphs of the partitions might be disconnected. But we can show that it includes all the homomorphisms that \emph{matter}, that is, the ones that can be extended to valid homomorphisms of the entire pattern $H$.

The proof follows a similar structure to the original algorithm from Bressan (see Section $3$ in \cite{Br21}).
We start by introducing the following lemma.

\begin{lemma} [Lemma $4$ in \cite{Br21}, restated] \label{lem:partialbressan}
    Let $\vec{H},\vec{G}$ be two directed acyclic graphs. Given any $B \subseteq S(\vec{H})$, the set of homomorphisms $\Phi(\vec{H}[B], \vec{G})$ has size $poly(d)n^{|B|}$ and can be enumerated in time $poly(d)n^{|B|}$. Where $k = |V(\vec{H})|$, $n=|V(\vec{G})|$ and $d = \Delta^+(\vec{G})$.
\end{lemma}

Like we mentioned previously, we can filter the homomorphisms of the skeletons of each partition to obtain a subset of the DAH homomorphisms of the induced subhypergraphs of the same partition.

\begin{lemma} \label{lem:partialhyper}
    Let $\vec{H},\vec{G}$ be two directed acyclic hypergraphs and $0 \leq \lparam \leq k = |V(\vec{H})|$. Given any $B \subseteq S(\vec{H})$, we can compute a set of homomorphisms $\Phi' \subseteq \Phi(\vec{H} [B]_\lparam,\vec{G})$ such that every $\phi \in \Phi(\vec{H} [B]_\lparam,\vec{G}) \setminus \Phi'$ has $\ext{\phi}{\vec{H}}{\vec{G}} = 0$ in time $poly(d_{\lparam})n^{|B|}+O(m)$. Moreover, $|\Phi'| = poly(d_{\lparam})n^{|B|}$. Where $d_{\lparam}$ is the maximum $l$-outdegree of $\vec{G}$.
\end{lemma}
\begin{proof}
    First we compute the $\lparam$-skeletons of $\vec{H}$ and $\vec{G}$, this can be done in $O(m)$ time. Note that $\Delta^+(\Skeleton{\lparam}{\vec{G}}) = d_{\lparam}$. We use \Lem{partialbressan} to compute  $\Phi(\Skeleton{\lparam}{\vec{H}}[B], \vec{G})$ in time $poly(d_\lparam) n^{|B|}$. We construct $\Phi'$ by filtering every homomorphism and ensuring that they are a valid homomorphism from $\vec{H}[B]_l$ to $\vec{G}$, this takes $O(1)$ time for homomorphism. The entire process will take $poly(d_{\lparam})n^{|B|}+m$ time.
    
    We just need to show that every homomorphism in $\Phi(\vec{H} [B]_\lparam,\vec{G}) \setminus \Phi'$ has extension equal to $0$. Assume otherwise, there is a homomorphism $\phi' \in \Phi(\vec{H} [B]_\lparam,\vec{G}) \setminus \Phi'$ with $\ext{\phi'}{\vec{H}}{\vec{G}} \not= 0$. Let $\phi$ be a valid homomorphism from $\vec{H}$ to $\vec{G}$ that extends $\phi'$, $\phi$ must also be a valid homomorphism from $\Skeleton{\lparam}{\vec{H}}$ to $\Skeleton{\lparam}{\vec{G}}$, and therefore its restriction to $\vec{H} [B]_\lparam$, $\phi'$ must also be a valid homomorphism of the skeleton. But in that case ,we would have that $\phi'\in \Phi'$, reaching a contradiction.
    
\end{proof}

We must also adapt Lemma $3$ in \cite{Br21} to hypergraph homomorphisms. Recall $\down(B)$ is the down-closure of $B$ in $\cT$.

\begin{lemma} \label{lem:hom_product}
    Let $\lparam \in \mathbb{Z}^+\cup\{\infty\}$ and let $\cT$ be a \dagtree{} of $\Skeleton{\lparam}{\vec{H}}$ and $B$ a node of $\cT$. Fix $\phi : \vec{H}[B]_\lparam \to \vec{G}$. For every child $B'$ of $B$ in $\cT$, let $\phi_{B'}$ be the restriction of $\phi$ to the vertices of $\Reach_{\vec{H},l}(B)\cap \Reach_{\vec{H},l}(\down(B'))$.

    \[
        \ext{\phi}{\vec{H}[\down(B)]_\lparam}{\vec{G}} = \prod_{B' \in children(B)} \ext{\phi_{B'}}{\vec{H}[\down(B')]_\lparam}{\vec{G}}
    \]
\end{lemma}
\begin{proof}
    Let $\Phi(\phi)$ be the set of homomorphisms from $\vec{H}[\down(B)]_\lparam$ to $\vec{G}$ that extend to $\phi$. Similarly, let $\Phi_{B'}(\phi_{B'})$ be the set of homomorphisms from $\vec{H}[\down(B')]_\lparam$ to $\vec{G}$ that extend $\phi_{B'}$. We can show that there is a bijection between $\Phi(\phi)$ and $\bigtimes_{B' \in\ children(B)} \Phi_{B'}(\phi_{B'})$.

    First, for any $\phi' \in \Phi(\phi)$ let $\phi'_{B'}$ be the restriction of $\phi'$ to $\vec{H}[\down(B')]_\lparam$. Note that $\phi'_{B'}$ must agree with $\phi_{B'}$, therefore $\phi'_{B'} \in \Phi_{B'}(\phi_{B'})$. Then we will have that the tuple $(\phi'_{B'}:B' \in\ children(B)) \in \bigtimes_{B' \in\ children(B)} \Phi_{B'}(\phi_{B'})$.

    Conversely, let $(\phi'_{B'}:B' \in children(B)) \in \bigtimes_{B' \in\ children(B)} \Phi_{B'}(\phi_{B'})$, we show that the homomorphism  $\phi'$ obtained by combining $\phi$ and $\phi'_{B'}$ for all $B' \in\ children(B)$ is in $\Phi(\phi)$. First, note that every homomorphism $\phi'_{B'}$ must agree with each other, since they can only intersect on the vertices of $\vec{H}[B]_\lparam$, but the values of those vertices for all homomorphism will correspond to $\phi$. Second, $\phi'$ is a mapping from $\vec{H}[\down(B)]_\lparam$ to $\vec{G}$ since every vertex in $\vec{H}[\down(B)]_\lparam$ is in $\vec{H}[B]_\lparam$ or $\vec{H}[\down(B')]_\lparam$ for some child $B'$ of $B$. Finally, we can show that every hyperedge in $\vec{H}[\down(B)]_\lparam$ is properly maintained in the homomorphism.

    Let $e$ be a hyperedge in $\vec{H}[\down(B)]_\lparam$ and $s$ its source (first vertex in its internal ordering). We can distinguish two cases:
    \begin{itemize}
        \item $s \in \Reach_{\vec{H},\lparam}(B)$, in this case $e \subseteq \Reach_{\vec{H},\lparam}(B)$, and therefore $\phi'(e) = \phi(e) \in E(\vec{G})$.
        \item $s \in \Reach_{\vec{H},\lparam}(\down(B'))$ for some $B' \in\ child(B)$, in this case $e \subseteq \Reach_{\vec{H},\lparam}(\down(B'))$, and therefore $\phi'(e)=\phi'_{B'}(e) \in E(\vec{G})$.
    \end{itemize}

    Therefore, $\phi' \in \Phi(\phi)$. 
\end{proof}

We now introduce the following algorithm, which is almost identical to Algorithm $1$ in \cite{Br21} (Bressan's algorithm), the main difference being that we only ensure that homomorphisms with non-zero extension get counted accurately.

\begin{algorithm}
\caption{$HomCount_l(\vec{H},\vec{G},B,\cT)$}
\begin{algorithmic}[1]
\State Compute $\Phi'$ as in \Lem{partialhyper}.
\State Initialize $C_B$ as an empty dictionary with default value of $0$.
\If{$B$ is a leaf of $\cT$}
    \For{every homomorphism $\phi \in \Phi'$}
        \State $C_B(\phi) = 1$
    \EndFor
\Else
    \For{every child $B'$ of $B$ in $\cT$}
        \State $C_{B'} = HomCount(\vec{H},\vec{G},B',\cT)$
        \State Initialize $AGG_{B'}$ as an empty dictionary with default value of $0$.
        \For {every nonzero entry $\phi$ in $C_{B'}$}
            \State Let $\phi_r$ be the restriction of $\phi$ to $\Reach_{\vec{H},l}(B)\cap \Reach_{\vec{H},l}(\down(B'))$
            \State $AGG_{B'}(\phi_r) += C_{B'}(\phi)$.
        \EndFor
    \EndFor
    \For{every homomorphism $\phi \in \Phi'$}
        \For{every child $B'$ of $B$ in $\cT$}
            \State Let $\phi_{B'}$ be the restriction of $\phi$ to $\Reach_{\vec{H},l}(B)\cap \Reach_{\vec{H},l}(\down(B'))$.
        \EndFor
        \State $C_B(\phi) = \prod_{B' \in \text{children}(B)}$ $AGG_{B'}(\phi_{B'})$
    \EndFor
\EndIf
\State Return $C_B$
\end{algorithmic}
\end{algorithm}

\begin{lemma}
    Let $\vec{H}$,$\vec{G}$ be any hypergraphs, $\cT=(\cB,\cE)$ a $l$-\dagtree{} of $\vec{H}$ and $B \in \cB$ any vertex from $\cT$. Algorithm $HomCount_l(\vec{H},\vec{G},B,\cT)$ returns a dictionary that for all homomorphisms $\phi: H[B]_\lparam \to \vec{G}$:
    \begin{itemize}
        \item If $\ext{\phi}{\vec{H}}{\vec{G}} > 0$, then $C_B(\phi) = \ext{\phi}{\vec{H}[\down(B)]_\lparam}{G)}$. 
        \item If $\ext{\phi}{\vec{H}}{\vec{G}} = 0$, then $C_B(\phi) \leq \ext{\phi}{\vec{H}[\down(B)]_\lparam}{G)}$.
    \end{itemize}

    Moreover, the algorithm runs in time $poly(d_\lparam)n^{\ldtw{l}(\cT)}\log{n}+O(m)$. Where $d_\lparam$ is the maximum $l$-outdegree of $\vec{G}$.
\end{lemma}
\begin{proof}
    The runtime analysis is similar to the one from Lemma $5$ in \cite{Br21}, as the algorithm takes the same steps, we ignore dependencies on the size of $H$. We will need to run the algorithm recursively for each bag in $\cT$, the construction of $\cT$ depends only on $H$ and therefore we will have $O(1)$ bags. Computing $\Phi'$ takes $poly(d_{\lparam})n^{|B|}+O(m)$ from \Lem{partialhyper}, and we will have $|\Phi'| = poly(d_{\lparam})n^{|B|}$. For each element in $\Phi$ we will require $O(\log{n})$ cost to access and write in the dictionary. Therefore, the complexity per iteration is $poly(d_{\lparam})n^{|B|}\log{n}+O(m)$, additionally it $|B|$ is bounded by $\ldtw{l}(\cT)$ for all $B$, therefore the algorithm will take $poly(d_{\lparam})n^{\ldtw{l}(\cT)}\log{n}+O(m)$ total time.
    
    We prove the correctness by induction on the nodes of the tree $\cT$.

    \textbf{Base case: }If $B$ is a leaf then $\vec{H}[B]_\lparam = \vec{H}[\down(B)]_\lparam$, note that from \Lem{partialhyper} we have that every homomorphism $\phi: \vec{H}[B]_\lparam\to \vec{G}$ with $\ext{\phi}{\vec{H}}{\vec{G}} > 0$ will be in $\Phi'$ and therefore the algorithm will set $C_B(\phi) = 1$ which means $C_B(\phi) = \ext{\phi}{\vec{H}[\down(B)\_\lparam}{\vec{G}} = 1$. If $\ext{\phi}{\vec{H}}{\vec{G}} = 0$ then either $\phi \in \Phi'$ and then $C_B(\phi) = \ext{\phi}{\vec{H}[\down(B)]_\lparam}{\vec{G}} = 1$, or $\phi \not\in \Phi'$ and then $C_B(\phi)$ just takes the default value of $0 < \ext{\phi}{\vec{H}[\down(B)]_\lparam}{\vec{G}}$.

    \textbf{Inductive step: }For the inductive step, assume that for every child $B'$ of $B$ in $\cT$ we have $C_{B'}(\phi) = \ext{\phi}{\vec{H}[\down(B')]_\lparam}{\vec{G}}$ for all $\phi \in \Phi(\vec{H}[B']_\lparam,\vec{G})$ with $\ext{\phi}{\vec{H}}{\vec{G}} > 0$ and $C_{B'}(\phi) \leq \ext{\phi}{\vec{H}[\down(B')]_\lparam}{\vec{G}}$ otherwise. Let $\phi$ be a homomorphism in $\Phi(\vec{H}[B]_\lparam,\vec{G})$ with $\ext{\phi}{\vec{H}}{\vec{G}}\neq 0$, we can show that $C_B(\phi) = \ext{\phi}{\vec{H}[\down(B)]_\lparam}{G)}$.
    
    Let $B'$ be a child of $B$, and $\phi_{B'}$ the restriction of $\phi$ to $\Reach_{\vec{H},l}(B)\cap \Reach_{\vec{H},l}(\down(B'))$, we can show that every homomorphism $\phi' \in \Phi(\vec{H}[B']_\lparam,\vec{G})$ with $C_{B'}(\phi') \neq 0$ that agrees with $\phi_{B'}$ will also have non-zero extension, $\ext{\phi'}{\vec{H}}{\vec{G}}\neq 0$. Let $\phi''$ be a homomorphism from $\vec{H}[\down(B')]_\lparam$ to $\vec{G}$ that agrees with $\phi'$, such homomorphism exists as $C_{B'}(\phi') \neq 0$. Let $\phi^*$ be a homomorphism from $\vec{H}$ to $\vec{G}$ that agrees with $\phi$, such homomorphism exists as $\phi$ has non-zero extension, and $\phi^*_{B'}$ its restriction to $V(\vec{H})\setminus (\Reach_{\vec{H},l}(\down(B'))\setminus \Reach_{\vec{H},l}(B))$.

    Note that $\phi^*_{B'}$ contains all the vertices in $\Reach_{\vec{H},l}(B))$, and therefore agrees with $\phi$ and $\phi_{B'}$. Similarly, $\phi''$ contains all the vertices in $(\Reach_{\vec{H},l}(\down(B'))$, and agrees with $\phi'$, which in turn agrees with $\phi_{B'}$. Let $V(\phi)$ be the set of vertices mapped by the homomorphism $\phi$, we have the following:

    \begin{align} \label{eq:union_hom}
        V(\phi^*_{B'}) \cup V(\phi'') = (V(\vec{H})\setminus (\Reach_{\vec{H},l}(\down(B'))\setminus \Reach_{\vec{H},l}(B)) ) \cup  (\Reach_{\vec{H},l}(\down(B'))= V(\vec{H})
    \end{align}
    \begin{align}
        V(\phi^*_{B'}) \cap V(\phi'') \nonumber
        &= (V(\vec{H})\setminus (\Reach_{\vec{H},l}(\down(B'))\setminus \Reach_{\vec{H},l}(B)) ) \cap  (\Reach_{\vec{H},l}(\down(B')) 
        \\&= \Reach_{\vec{H},l}(\down(B'))\setminus (\Reach_{\vec{H},l}(\down(B')) \setminus (\Reach_{\vec{H},l}(\down(B')) \cap \Reach_{\vec{H},l}(B)))
        \\&= \Reach_{\vec{H},l}(\down(B')) \cap \Reach_{\vec{H},l}(B) 
        = V(\phi_{B'}) \nonumber
    \end{align}
        
    Note that $\phi^*_{B'}$ and $\phi''$ intersect only in the vertices of $\phi_{B'}$, and we saw that both agree with $\phi_{B'}$. Therefore, we can combine $\phi^*_{B'}$ and $\phi''$ into a homomorphism $\phi^{T} = \phi''+\phi^*_{B'}$, and we can show that it must be a valid homomorphism from $\vec{H}$ to $\vec{G}$. First, from \Eqn{union_hom} we have that $V(\phi^T) = V(\phi^*_{B'}) \cup V(\phi'') = V(\vec{H})$. Only left to check is that every hyperedge is preserved in the image. Let $e$ be a hyperedge in $\vec{H}$, and let $s$ be the ``source'' of $e$, that is, the first vertex of its internal ordering. We distinguish two cases:
    \begin{itemize}
        \item $s \in \Reach_{\vec{H},l}(\down(B'))$, then $e \subseteq V(\phi'')$, and $\phi^T(e)=\phi''(e) \in E(\vec{G})$.
        \item $s \not\in \Reach_{\vec{H},l}(\down(B'))$, then for every vertex $v \in e$ we have $v \in \Reach_{\vec{H},l}(B)$, and therefore $v \in V(\phi^*_{B'})$; or $v \not\in \Reach_{\vec{H},l}(B)$, but then $v \not\in \Reach_{\vec{H},l}(\down(B'))$ as otherwise it would not be a valid $l$-\dagtree{}, which also means $v \in V(\phi^*_{B'})$. Thus, $e \subseteq V(\phi^*_{B'})$ and $\phi^T(e)=\phi^*_{B'}(e) \in E(\vec{G})$.
    \end{itemize}
    Therefore, $\phi^T$ is a valid homomorphism, which means $\ext{\phi'}{\vec{H}}{\vec{G}}\neq 0$. Thus, from the inductive hypothesis we will have that $C_{B'}(\phi') = \ext{\phi'}{\vec{H}[\down(B')]_\lparam}{\vec{G}}$. We will then have that:

    \begin{equation} \label{eq:agg}
        AGG_{B'}(\phi_{B'}) = \sum_{\substack{\phi' \in C_{B'} \\ \phi' \text{ agrees } \phi_{B'}}} C_{B'}(\phi') = \sum_{\substack{\phi': \vec{H}[B']_\lparam \to \vec{G} \\ \phi' \text{ agrees } \phi_{B'}}}  \ext{\phi'}{\vec{H}[\down(B')]_\lparam}{\vec{G}} = \ext{\phi_{B'}}{\vec{H}[\down(B')]_\lparam}{\vec{G}}
    \end{equation}

    The first equality comes from lines $9$-$12$ in the algorithm, the second from the inductive hypothesis together with the previous explanation, and the third from the fact that every homomorphism contributing to $\ext{\phi_{B'}}{\vec{H}[\down(B')]}{\vec{G}}$ also contributes to exactly one of the $\ext{\phi'}{\vec{H}[\down(B')]}{\vec{G}}$.
    Therefore, for every $\phi$ with non-zero extension, we will have:

    \begin{equation} \label{eq:cb}
        C_B(\phi) = \prod_{B' \in children(B)} \ext{\phi_{B'}}{\vec{H}[\down(B')]_\lparam}{\vec{G}} = \ext{\phi}{\vec{H}[\down(B)]_\lparam}{\vec{G}}
    \end{equation}

    Here, the first equality comes from combining line $16$ in the algorithm and \Eqn{agg}, and the second equality comes from \Lem{hom_product}.

    Now, let $\phi$ be a homomorphism in $\Phi(\vec{H}[B]_\lparam,\vec{G})$ with $\ext{\phi}{\vec{H}}{\vec{G}} = 0$, we can show that $C_B(\phi) \leq \ext{\phi}{\vec{H}[\down(B)]_\lparam}{G)}$. We can rewrite \Eqn{agg} by replacing the second equality with an inequality giving:
    \begin{equation}
        AGG_{B'}(\phi_{B'}) \leq \ext{\phi_{B'}}{\vec{H}[\down(B')]_\lparam}{\vec{G}}
    \end{equation}

    Using this instead in \Eqn{cb} will finally give:

    \begin{equation}
        C_B(\phi) \leq \ext{\phi}{\vec{H}[\down(B)]_\lparam}{\vec{G}}
    \end{equation}
    Completing the induction.
\end{proof} 

\subsection{Completing the proof}

We finally have all the pieces for completing the proof of \Thm{computing}, which we restate.

\computing*
\begin{proof}
	We start by computing the $l$-degeneracy ordering of $G$, and then create the directed acyclic hypergraph $\vec{G}^{(l)}$. From \Lem{degeneracy_ordering}, this can be done in time $O(n+m(\log{m} \cdot r(G) + r(G)^2)) = O(n+m\log{m})$ as the rank of $G$ is bounded. Moreover, we have that the maximum $l$-outdegree of $\vec{G}^{(l)}$ is $\Delta^+_\lparam({\vec{G}^{(l)}}) \leq \ldegen{\lparam}(G) \cdot r(G) = O(\ldegen{\lparam}(G))$.

    Using \Lem{contracted} and equation \Eqn{undirect_to_direct}, we can compute $\Hom{G}{H}$ by computing the value of $\Hom{\vec{G}}{\vec{H}}$ for all the directed acyclic hypergraphs in $\bigcup_{H' \in \contract(H)}\Sigma(H')$. The number of these instances is bounded by some function of $k$. Also, using \Lem{dtw_bound} we have that for any $H' \in \contract(H)$, $\ldtw{l}(H') \leq \ldtw{l}(H)$. And therefore, each $\vec{H}$ for which we need to compute $\Hom{\vec{G}}{\vec{H}}$, will also have $\ldtw{l}(\vec{H}) \leq \ldtw{l}(H)$.

    Finally, for each $\vec{H}$, we compute a $l$-\dagtree{} $\cT$ and root it arbitrarily in some node $R$. We then run $HomCount(\vec{H},\vec{G}^{(l)}),R,\cT)$, because $R$ is the root of $\cT$, we will have that the output $C_R$ is a dictionary such that for each homomorphism $\phi: \vec{H}[R]_\lparam\to \vec{G}$:
    \begin{itemize}
        \item if $\ext{\phi}{\vec{H}}{\vec{G}}>0$, then $C_R(\phi) = \ext{\phi}{\vec{H}[\down(R)]_\lparam}{\vec{G}} = \ext{\phi}{\vec{H}}{\vec{G}}$.
        \item if $\ext{\phi}{\vec{H}}{\vec{G}}=0$, then $C_R(\phi) \leq \ext{\phi}{\vec{H}[\down(R)]_\lparam}{\vec{G}} = \ext{\phi}{\vec{H}}{\vec{G}} = 0$.
    \end{itemize}
    Therefore, summing over all entries of $C_R$ we will obtain
    \[
        \sum_{\phi \in C_R} C_R(\phi) = \sum_{\phi: \vec{H}[R]_\lparam\to \vec{G}} \ext{\phi}{\vec{H}}{\vec{G}} = \Hom{\vec{G}}{\vec{H}}
    \]

     Running $HomCount(\vec{H},\vec{G}^{(l)}),R,\cT)$ will take $poly(d_\lparam)n^{\ldtw{l}(\cT)}\log{n}+O(m)$, we will have $\ldtw{\lparam}(\cT) \leq \ldtw{l}(\vec{H}) \leq \ldtw{l}(H)$, and $d_\lparam=O(\ldegen{\lparam}(G))$. Moreover, $m = O(\ldegen{\lparam}(G)\cdot n)$. Therefore, we will need $poly(\ldegen{\lparam}(G)) O(n^{\ldtw{l}(H)}\log{m})$ total time to compute $\Hom{G}{H}$.
    
 \end{proof}

\section{A characterization of patterns with $\dtw_l = 1$} \label{sec:characterize}

In this section, we prove \Thm{characterize}.

We start by giving an alternative definition of the sets $\setH{\lparam}$, this definition will help us prove the theorem. First, we define an \connector{l}.

\begin{definition}[\connector{l}] \label{def:connector}
	Let $H$ be a hypergraph. Let $X$ and $V$ be two disjoint sets of vertices of $H$, such that $X\cup V = V(H)$. We say that $H$ is an \connector{l} of $V$ if:
	\begin{itemize}
		\item $H$ is connected.
		\item Every hyperedge of $H$ with at least two vertices in $X$ has at least $l+1$ vertices in $V$. $\forall e \in E(H') | |e \cap X|\geq 2, |e \cap V| > l$.
		\item Let $H'$ be the trimmed subhypergraph of $H$ induced by $V$, there exists at least one connected component $V' \subset V$ in $H'$, such that every vertex in $X$ is a neighbor (shares a hyperedge in $H$) with a vertex in $V'$.
	\end{itemize}
\end{definition}

Secondly, we define the sets $\setHk{l}{k}$. Basically, a hypergraph belongs to $\setHk{l}{k}$ if it can be split into a core $V$ of $k$ vertices and $k$ \connector{l}s, such that every subset of $k-1$ vertices has its own \connector{l}.

\begin{definition} [$\setHk{l}{k}$] \label{def:seth}
	We define $\setHk{l}{k}$ as the set of all hypergraphs $H$ such that we can split the vertices of $V(H)$ into the set $X=\{x_0,..,x_{k-1}\}$ and $k$ disjoint sets $V_0,...,V_{k-1}$, with $X \cup \bigcup_i V_{i} = V(H)$; and split the set of hyperedges $E(H)$ into $k$ disjoint sets $E_0,...,E_{k-1}$, with $\bigcup_i E_{i} = E(H)$; satisfying:
	\begin{itemize}
		\item For all $i \in [0,k-1]$, the hypergraph $H_i=(V_i\cup (X\setminus\{x_i\},E_i)$ is an \connector{l} of $X \setminus \{x_i\}$.
		\item For all $i$, all hyperedges $e \in E_i$ are completely contained in $V_i\cup (X\setminus\{x_i\})$. Equivalently, no hyperedge $e \in E(H)$ contains two vertices in different subsets $V_i,V_j$, or connects $V_i$ and the vertex $x_i$.
	\end{itemize}
\end{definition}

We can show that \Def{obstruct} and \Def{seth} are equivalent.

\begin{lemma}
	\[
		\setH{\lparam} = \bigcup_{k\geq 3} \setHk{\lparam}{k}
	\]
\end{lemma}
\begin{proof}
	Let $H$ be a pattern in $\setH{\lparam}$ with $|C| = k$, we show that it is also in $\setHk{\lparam}{k}$. Let $C,\{D_1,D_2...\},\{E'_1,E'_2,...\}$ be the partitions of $H$ according to \Def{obstruct}.
	 We can set $X = C$, and for all $i \in [k]$ we set $V_i = D_{i+1}$, $E_i =E'_{i+1}$. For every $i > k$ we find some $x_j$ not in $E_i$ (one must exist by definition) and add $D_i$ to $V_j$ and $E'_i$ to $E_i$. We can show that the result is a valid partition of $H$ according to \Def{seth}:
	 \begin{itemize}
	 	\item $H_i=(V_i\cup (X\setminus\{x_i\},E_i)$ is an \connector{l} of $X \setminus \{x_i\}$:
	 	\begin{itemize}
	 		\item Condition $1$ is satisfied as all the connected components in $E_i$ touch some vertex in $X \setminus \{x_i\}$, and the set $E'_{i+1}$ connects all such vertices.
	 		\item Condition $2$ is satisfied, as the induced $l$-trimmed subhypergraph of the core is empty (or only contains singleton hyperedges).
	 		\item Condition $3$ is satisfied, again by $E'_{i+1}$.
	 	\end{itemize}
 		\item Every hyperedge in $E'_{i+1}$ was contained in $D_{i+1}\cup (C\setminus\{c_{i+1}\})$, and therefore every hyperedge $E_i$ will be contained in $V_i\cup (X\setminus\{x_i\})$.
	 \end{itemize}
 
 	Now, let $H$ bet a pattern in $\setHk{\lparam}{k}$, we show that it is also in $\setH{\lparam}$, let $X$,$\{V_0,...,V_{k-1}\}$,$\{E'_0,...,E'_{k-1}\}$ be the sets forming a valid partition of $H$ according to \Def{seth}. Set $C= X$, let $H'$ be the induced trimmed subhypergraph of $H\setminus C$, we will have connected components $D_i$, note that a $D_i$ cannot have vertices originally in different sets $V_j$, similarly the sets containing $E_i$ cannot be part of different sets $E'_j$. We prove that this will be a valid partition according to \Def{obstruct}:
 	\begin{enumerate}
 		\item Every hyperedge is part of some \connector{\lparam}, and therefore, every hyperedge with $2$ or more vertices in the core $C$ must have at least $l+1$ vertices outside of it, hence, it will not appear in the $l$-trimmed subhypergraph induced by the core.
 		\item This is true from the second condition of \Def{seth}.
 		\item We can reorder the sets $E_i$ to satisfy this. Note that every set $C\setminus \{c_i\}$ forms an \connector{\lparam}, which requires some connected component (which will be a set $D_i$) to neighbor all the vertices in $C\setminus \{c_i\}$, which implies that $E_i$ contains $C\setminus \{c_i\}$.
 	\end{enumerate}
\end{proof}

A hypergraph $H$ having $\ldtw{\lparam}(H)>1$ is equivalent to the reachability hypergraph of the $l$-skeletons of one of its acyclic orientations $\vec{H}$ not being $\alpha$-acyclic. In order to capture the $\alpha$-acyclicity of the reachability graph, we introduce the following definitions.

\begin{definition} [\reachcycle{k}]
    Let $\vec{H}$ be a directed acyclic graph and $k\geq 3$, a \reachcycle{k} is a pair $X,Y$, where $Y=\{y_0,...,y_{k-1}\} \subset V(\vec{H})$ and $X=\{x_0,...,x_{k-1}\}\subset V(\vec{H})$, such that for every $i\in [0,k-1]$, $\{x_{i},x_{i+1}\} \subseteq \Reach(y_i)$ (taken modulo $k$ in the indexes) and every vertex $v$ with $|\Reach{v}\cap X| \geq 2$ must have $\Reach{v}\cap X = \{x_{i},x_{i+1}\}$ for some $i \in [0,k-1]$.
\end{definition}

\begin{definition} [\reachsimplex{k}] 
    Let $\vec{H}$ be a directed acyclic graph and $k\geq 3$, a \reachsimplex{k} is a pair $X,Y$, where $Y=\{y_0,...,y_{k-1}\} \subset V(\vec{H})$ and $X=\{x_0,...,x_{k-1}\}\subset V(\vec{H})$, such that for every $i \in [0,k-1]$, $\{ X \setminus \{x_i\}\} \subseteq \Reach(y_i)$ and no vertex $v \in V(\vec{H})$ exists with $X \subseteq \Reach(v)$.
\end{definition}

Note that a \reachcycle{3} and a \reachsimplex{3} are equivalent. 

We can relate the previous definitions to the concept of $\alpha$-acyclicity. The following theorem proves that having a reachability hypergraph that is not $\alpha$-acyclic is equivalent to containing either a \reachcycle{k} or a \reachsimplex{k}.

\begin{theorem} (Theorem $8$ in \cite{BeGiLe+22}, originally in \cite{BeFaMa+81,BeFaMa+83}) \label{thm:acyclic}
    A hypergraph $F$ is $\alpha$-acyclic if and only if for every $k \geq 3$, there is no $S=\{x_0,x_1,...,x_{k-1}\}\subseteq V(F)$ such that one of the following conditions holds:
\begin{enumerate}
    \item For every $0 \leq i \leq k-1$ there exists a hyperedge $e\in E(F)$ such that $e \cap S = \{x_i,x_{i+1}\}$, and there is no $e \in E(F)$ with $|e \cap S|\geq 2$ such that $e \cap S \neq \{x_i,x_{i+1}\}$ for all $- \leq i \leq k-1$. (Indices taken modulo $k$.)
    \item For every $0 \leq i \leq k-1$ there exists $e \in E(F)$ such that $e \cap S = S \setminus \{x_i\}$, and there is no edge $e\in E(F)$ such that $S\subseteq e$.
\end{enumerate}
\end{theorem}

We now relate containing a \reachcycle{k} in the $l$-skeleton of an orientation $\vec{H}$ with $H$ having a hypergraph in $\setHk{l}{3}$ as an induced trimmed subhypergraph. First, we introduce the following definition.

\begin{definition} [First \reachcycleempty{}]
     Let $\ord: V(\vec{H}) \to \NN$ be an ordering of the vertices of $\vec{H}$, and for a set of vertices $V$, let $\ord(V)_{(i)}$ denote the ordering of the $i$-th latest vertex in $V$ with respect to $V$. A \reachcycle{l} $X,Y$ \textbf{precedes} a \reachcycle{l'} $X',Y'$ in the ordering $\ord$ if either $X \subset X'$ or there exists an $j\leq \min{(|X|,|X'|)}$ such that:
    \begin{itemize}
        \item For all $i < j$ the vertices with the $i$-th latest ordering in $X$ and $X'$ are the same, $\forall i< j, \ord(X)_{(i)} = \ord(X')_{(i)}$.
        \item The vertex with the $j$-th latest ordering in $X$ precedes the vertex with the $j$-th latest ordering in $X'$, $\ord(X)_{(j)} < \ord(X')_{(j)}$.
    \end{itemize}
    We call $X,Y$ the First \reachcycleempty{} of $\vec{H}$ if it precedes all the other \reachcycleempty{s} in $\vec{H}$.
\end{definition}

\begin{lemma} \label{lem:cycle_to_seth}
    Let $\vec{H}$ be a directed acyclic orientation of the hypergraph $H$. If $\Skeleton{l}{\vec{H}}$ contains a \reachcycle{k} then $H$ contains a hypergraph $H' \in \setHk{l}{3}$ as an induced trimmed subhypergraph.
\end{lemma}
\begin{proof}
    Let $\ord: V(H) \to \NN $ be an ordering of the vertices of $H$ consistent with $\vec{H}$. And let $X,Y \subseteq V(H)$ be the first \reachcycleempty{} in $\Skeleton{l}{\vec{H}}$.
    
    Let $k = |X| = |Y|$, for every $i \in [0,k-1]$, let $V_i$ be the set that contains, for every $v \in V(H)$ such that $\{x_i,x_{i+1}\} \subseteq \Reach_{\Skeleton{l}{\vec{H}}}(v)$, $v$ and all the vertices in the shortest path from $v$ to $x_i$ and $x_{i+1}$ in $\Skeleton{l}{\vec{H}}$ (not including $x_i$ and $x_{i+1}$ themselves).
    Let $H_i$ be the trimmed subhypergraph of $H$ induced by $V_i \cup \{x_i,x_{i+1}\}$ and $E_i = E(H_i)$.

    We show that $H_i$ must be an \connector{l} of $\{x_i,x_{i+1}\}$. First, note that $V_i$ is not empty as $\{x_i,x_{i+1}\} \subseteq \Reach_{\Skeleton{l}{\vec{H}}}(y_i)$, and therefore $y_i \in V_i$. 
    \begin{itemize}
        \item We can see that $H_i$ is connected, every vertex is part of some path in $\Skeleton{l}{\vec{H}}$, every pair connects a vertex $v$ to either $x_i$ or $x_{i+1}$, and every vertex $v$ has two such paths; therefore, all the paths will form a single connected component. Every direct edge $u,v$ in $\Skeleton{l}{\vec{H}}$ corresponds to some hyperedge $e$ in $H$ that will appear partially in $H_i$ connecting $u$ and $v$.
        \item Let $e$ be a hyperedge in $H_i$ containing both $\{x_i,x_{i+1}\}$,  $\Skeleton{l}{\vec{H}}$ does not contain a direct edge between $x_i$ and $x_{i+1}$ (or vice-versa) which means that there are at least $l$ vertices in the internal ordering of $e$ before both $x_{i}$ and $x_{i+1}$, note that the first $l+1$ vertices must connect to both $x_{i}$ and $x_{i+1}$ in $\Skeleton{l}{\vec{H}}$ and therefore they will be in $V_i$, thus $|e \cap V_i|> l$.
        \item Finally, let $H'$ be the trimmed subhypergraph of $H$ induced by $V_i$, note that every connected component in $H'$ will have to connect to both $x_{i}$ and $x_{i+1}$ as it necessarily contains a vertex $v$ with its direct paths to $x_i$ and $x_{i+1}$ in $\Skeleton{l}{\vec{H}}$.
    \end{itemize}

    We can also show the following claim.

    \begin{claim} \label{clm:neighbors}
        For every $i \neq j$, no vertex $v_i \in V_i$ shares an edge with a vertex $v_j \in V_j$.
    \end{claim}
    \begin{proof}
    Assume otherwise that there is a vertex $v_i \in V_i$ and a vertex $v_j \in V_j$ such that there exists a hyperedge $e \in E(H)$ with $v_i,v_j \in e$. Let $v'_i$ be a predecessor of $v_i$ (or $v_i$ itself) that reaches both $x_i$ and $x_{i+1}$ and similarly let $v'_j$ be a predecessor of $v_j$ (or $v'_j$ itself) that reaches both $x_j$ and $x_{j+1}$. Having a hyperedge connecting $v_i$ and $v_j$ in $H$ implies one of the following three scenarios in $\Skeleton{l}{\vec{H}}$. We show that they all lead to some contradiction:
    \begin{itemize}
        \item There is a direct edge connecting $v_i$ to $v_j$. If $v_j$ reaches a vertex in $X$ other than $x_i,x_{i+1}$ then $|\Reach(v'_i) \cap X |\geq 3$, which is not possible. Thus, we have that $v_j$ reaches either $x_i$ or $x_{i+1}$ (but not both). Without loss of generality, let $x_{i+1} \in \Reach(v_j)$ (which means $j=i+1$), set $X' = v_j \cup X \setminus \{x_{i+1}\}$ and $Y' = (Y\setminus \{y_i,y_{i+1}\})\cup\{v'_i,v'_j\}$.

        Note that $\ord(v_j) < \ord(v_{i+1})$ and therefore $X'$ precedes $X$. Moreover, we can verify that $X',Y'$ forms a \reachcycle{k} as no vertex can reach $3$ or more vertices in $X'$, since it will mean that it can reach the same number of vertices in $X$, leading to a contradiction.
        \item There is a direct edge connecting $v_j$ to $v_i$. This is symmetrical to the previous case.
        \item There is a vertex $v$ with direct edges connecting it to $v_i$ and $v_j$. We distinguish the following four sub-cases:
        \begin{itemize}
            \item $|X \cap \Reach(v_i)| = 2 $ and $|X \cap  \Reach(v_j)| = 2$. This implies $\{x_j,x_{j+1},x_i,x_{i+1}\} \subseteq \Reach(v)$ and therefore $v$ reaches at least $3$ vertices in $X$, which is not possible.
            \item $|X \cap  \Reach(v_i)| = 2$ and $|X \cap  \Reach(v_j)|= 1$. In this case $v_j$ must reach $x_i$ or $x_{i+1}$, otherwise $v$ reaches three vertices in $X$. Without loss of generality, let $x_{i+1}\in \Reach(v_j)$ (and $j=i+1$). Set $X'= \{v_j\} \cup X \setminus\{x_{i+1}\}$ and $Y' = (Y \setminus \{y_i,y_{i+1}) \cup \{v, v'_j\}$. $\ord(v_j) < \ord(x_{i+1})$, which means that $X'$ precedes $X$, and $X',Y'$ will form a valid \reachcycle{k}; therefore, $X,Y$ is not the first \reachcycleempty{}, leading to a contradiction.
            \item $|X \cap  \Reach(v_i)|=1$ and $|X \cap  \Reach(v_j)|=2$. Symmetrical to the previous case.
            \item $|X \cap  \Reach(v_i)|=1$ and $|X \cap  \Reach(v_j)|=1$. If $X \cap  \Reach(v_i) \neq X \cap  \Reach(v_j)$, then for some $i'$, $x_{i'} \in \Reach(v_i)$ and $x_{i'+1} \in \Reach(v_j)$ (or vice-versa, we assume the former without loss of generalization). There are two further possibilities, either $j = i'=i+1$ or $j=i'+1=i+2$. In the former, we can set $X'= \{v_i\} \cup X \setminus \{x_{i+1}\}$ and $Y'=(Y \setminus \{ y_i,y_{i+1}\})\cup \{ v'_i,v\}$ to obtain a valid \reachcycleempty{} where $X'$ precedes $X$ as $\ord(v_i) < \ord(x_{i+1})$. In the latter, we can set $X'=  \{v_i,v_j\} \cup X \setminus \{x_{i+1},x_{i+2}\}$ and $Y'= (Y \setminus \{ y_i,y_{i+1},y_{i+2}\})\cup \{ v'_{i},v,v'_j\}$ to obtain a valid \reachcycleempty{}, again $X'$ precedes $X$ as both $\ord(v_i) < \ord(x_{i+1})$ and $\ord(v_{j}) < \ord(x_{i+2})$.

            Alternatively, $X \cap  \Reach(v_i) = X \cap  \Reach(v_j)$, and they both reach $x_i$ or $x_{i+1}$, without loss of generality, assume $X \cap  \Reach(v_i)\cap  \Reach(v_j) = \{x_{i+1}\}$ (and $j=i+1$). Set $X' = \{v_i,v_j\} \cup X \setminus  \{x_{i+1}\}$ and $Y' = (Y\setminus\{y_i,y_{i+1}\})\cup\{v'_i,v,v'_j\}$. Both $\ord(v_i) < \ord(x_{i+1})$ and $\ord(v_j) < \ord(x_{i+1})$, which means that $X'$ precedes $X$. The only possibility that $X',Y'$ does not form a \reachcycleempty{} is that there is a vertex $u$ that reaches both $v_i,v_j$ and either $x_{i}$ or $x_{i+2}$, but then we can construct the sets $X''= X' \setminus \{v_i\}$ and $Y''= u \cup Y' \setminus \{v,v'_i\}$ in the former or $X'' = X' \setminus \{v_j\}$, $Y''= u \cup Y' \setminus \{v,v'_j\}$ in the latter which will be a valid \reachcycleempty{}. Moreover, $X'' \subset X'$, therefore, $X''$ precedes $X'$ and $X$, leading to a contradiction.
        \end{itemize}
    \end{itemize}
    \end{proof}
    
     We can show that $H'$, the trimmed subhypergraph of $H$ induced by $X \cup \bigcup_{i=0}^{k-1}V_i$ is in $\setHk{l}{3}$. If $k=3$, then we are done as every hypergraph $H_i$ is an \connector{l} and by \Clm{neighbors} they do not share any hyperedge. Otherwise, let $V'_2 = \bigcup_{i=2}^{k-1} V_i \cup \bigcup_{i=3}^{k-1}x_i$, and $H'_2$ be the trimmed subhypergraph of $H$ induced by $V'_2 \cup \{x_2,x_0\}$. We can verify that this split satisfies all the conditions of \Def{seth}:
    \begin{itemize}
        \item $H_0$ is an \connector{l} for $V\setminus \{x_2\} = \{x_0,x_1\}$.
        \item $H_1$ is an \connector{l} for $V\setminus \{x_0\} = \{x_1,x_2\}$.
        \item $H'_2$ is an \connector{l} for $V\setminus \{x_1\} = \{x_2,x_0\}$.We can verify the following conditions:
        \begin{itemize}
            \item $H'_2$ is connected, as every $V_i$ connects with the vertices $x_i$,$x_{i+1}$ which will be included in $H'_2$. It also forms a connected component if only including the vertices in $V'_2$.
            \item No hyperedge in $H'_2$ includes both $\{x_2,x_0\}$ as it will mean the existence of a vertex reaching both of them in $\Skeleton{l}{\vec{H}}$, which is not possible for $k>3$.
        \end{itemize}
        \item No hyperedge contains two vertices in different subsets, again this is true from \Clm{neighbors}.
    \end{itemize}
    Moreover, $V_0 \cup V_1 \cup V'_2 \cup \{x_0,x_1,x_2\} = V(H')$ and $E_0 \cup E_1 \cup E(H'_2) = E(H')$. Therefore, we will have that $H'\in \setHk{l}{3}$.
\end{proof}

We can prove a similar result for \reachsimplex{k}. We show that containing a \reachsimplex{k} in $\Skeleton{l}{\vec{H}}$ implies containing a pattern in $\setHk{l}{k}$ as an induced trimmed subhypergraph. We will first prove the following auxiliary lemma.

\begin{lemma} \label{lem:groups}
    Let $\vec{H}$ be a directed acyclic graph with ordering $\ord$ that does not contain a \reachcycle{k} for any value of $k$ but contains some \reachsimplex{k} for some $k$. Let $k'$ be the minimum $k$ such that $\vec{H}$ contains a \reachsimplex{k} and $X,Y$ be the sets that form a \reachsimplex{k'} with minimum total ordering $\sum_{v \in X} o(v)$. Let $V \subseteq V(\vec{H} \setminus X)$ be a connected subset of vertices such that $\forall v \in V$, $|\Reach(v) \cap X| \geq 1$. We will have $X \not\subseteq \Reach_{\vec{H}}(V)$.
\end{lemma}

\begin{proof} 
    We assume $X \subseteq \Reach_{\vec{H}}(V)$ and prove by contradiction. Let $V' = \{v \in V(\vec{H}) \setminus X | |\Reach(v) \cap X| \geq 1\}$ be the set of vertices of $V(\vec{H}) \setminus X$ that reach some vertex in $X$. Note that $V' \subseteq V''$. Let $\vec{H''}$ be the subgraph of $\vec{H}$ induced by $V''$, this graph might be disconnected, but at least one of its connected components must contain $V'$. Let $V^*$ be the set of vertices inducing that connected component. Clearly $V'\subseteq V^*$. We prove the following claim.

    \begin{claim} \label{clm:pair_vertices}
        Let $v_1,v_2$ be two vertices in $V^*$, there must exist a vertex $v_3 \in V^*$ such that $(\Reach(v_1)\cup \Reach(v_2) ) \cap X \subseteq \Reach(v_3) \cap X$.
    \end{claim}
    \begin{proof}
    If $\Reach(v_1) \subseteq \Reach(v_2)$, then $v_3 = v_2$ and we are done. Similarly if $\Reach(v_2) \subseteq \Reach(v_1)$. Otherwise, there is no direct path from $v_1$ to $v_2$. However, $v_1$ and $v_2$ are part of the same connected component in $\vec{H''}$ and there must be an undirected path connecting them. 
    
    An $l$-alternating path is a sequence $a_0,b_1,...,b_l,a_{l}$ such that $b_i \in \Reach(a_{i-1}) \cap \Reach(a_i)$ for all $i \in [1,l]$. An $l$-alternating path is minimal if there is no $l'$-alternating path connecting $a_0$ and $a_l$ for some $l'<l$. Let $l'$ be the minimum $l$ such that there is an $l$-alternating path with $a_0 = v_1$ and $a_{l} = v_2$ using only vertices in $V^*$. It suffices to show that for any $l$-alternating path, there exists a vertex $v$ with $((\Reach(a_0) \cup \Reach(a_l))\cap X) \cup \{b_i\}$ for some $1\leq i\leq l$.
    
    We use induction on the size $l$ of the alternating path to show that such a vertex must exist.

    \paragraph{Base case:} First, we show for $l=1$. We have a vertex $b$ such that $b \in \Reach(a_0) \cap \Reach(a_1)$. Let $\mathcal{Z}$ be the collection of possible subsets of $((\Reach(a_0) \cup \Reach(a_1)) \setminus \Reach(b)) \cap X$ such that there exists a vertex $v$ that reaches both $b$ and the vertices of the subset. That is, $\mathcal{Z} = \{ Z \subseteq ((\Reach(a_0) \cup \Reach(a_1)) \setminus \Reach(b)) \cap X | \exists v \in V^*, Z \cup b \subseteq \Reach(v) \}$.

    If $\mathcal{Z}$ contains $((\Reach(a_0) \cup \Reach(a_1)) \setminus \Reach(b)) \cap X$ then there is a vertex $v \in V^*$ that reaches $b$ and all the vertices in $((\Reach(a_0) \cup \Reach(a_1)) \setminus \Reach(b)) \cap X$, but because $v$ reaches $b$ it will also reach all the vertices in $\Reach(b)$, hence $\Reach(v) \supseteq (\Reach(a_0) \cup \Reach(a_1) \cap X$.

    Else, let $Z'$ be the smallest subset of $((\Reach(a_0) \cup \Reach(a_1)) \setminus \Reach(b)) \cap X$ such that $Z' \not\in \mathcal{Z}$. Note that $|Z'|< k'$, as $|\Reach(b) \cap X|\geq 1$, also $|Z'|\geq 2$, as every vertex in $((\Reach(a_1) \cup \Reach(a_2)) \setminus \Reach(b)) \cap X$  is reached by either $a_0$ or $a_1$. We distinguish two cases:
    \begin{itemize}
        \item $|Z'| = 2$: Without loss of generality, let $x_0,x_1 \in X$ be the two vertices in $Z'$ and let $x_0 \in \Reach(a_0)$ and $x_1 \in \Reach(a_1)$. Recall that the vertices of $X$ are part of a \reachcycle{k'}, therefore, there is a vertex $y \in Y$ that reaches $x_0,x_1$. Note that the sets $Y'=\{a_0,a_1,y\}$ and $X'=\{x_0,x_1,b\}$ will form a \reachcycle{3}, as $\{x_0,x_1\} \in \Reach(y)$, $\{x_0,b\} \in \Reach(a_0)$,  $\{x_1,b\} \in \Reach(a_1)$, and no vertex reaches all vertices in $X'$. But this violates the assumption of the lemma that $\vec{H}$ does not contain any \reachcycle{k}, leading to a contradiction.
        \item $|Z'| \geq 3$: Let $\mathcal{Z}' = \{ Z \subset Z' | |Z| = |Z'| - 1\}$ and note that every set in $\mathcal{Z}'$ must also be in $\mathcal{Z}$, therefore for every set $Z \in \mathcal{Z}'$ there is a vertex $v_Z$ such that $b \cup Z \in \Reach{v_z}$; additionally, because $|Z'| < k'$ there is a vertex $y \in Y$ with $Z \in \Reach(y)$. We set $Y' = \{y\}\cup\{v_Z |Z \in \mathcal{Z}' \}$ and $X' = Z' \cup \{b\}$, let $k''=|X'|=|Y'|\leq k'$. 
        
        We can see that $Y',X'$ forms a \reachsimplex{k''} as no vertex reaches the entire set $X'$ (otherwise $Z'$ would be in $\mathcal{Z}$). If $k'' < k'$, then $k'$ is not the minimum $k$ such that $\vec{H}$ contains a \reachsimplex{k}, which leads to a contradiction. Otherwise, if $k'' = k'$ we can see that the only difference between $X'$ and $X$ is that $X'$ contains the vertex $b$ and $X$ the vertex $c = X\setminus Z'$, but in this case $c \in \Reach(b)$ which implies that $b$ is strictly before $c$ in any ordering of the vertices of $\vec{H}$, that is, $\ord(b) < \ord(c)$ and therefore $\sum_{y \in Y'} \ord(y) < \sum_{y\in Y} o(y)$ which would again contradict the assumption from the lemma.
    \end{itemize}
    
    \paragraph{Inductive step:} Assume that for any $l < l'$ we have that for all $l$ shortest alternating path with endpoints $a',b' \in V^*$ there is a vertex $v$ with $\Reach(a')\cup \Reach(b') \cup b_i \subseteq \Reach(v)$ for some $b_i$ in the alternating path. We show that the same holds for $l=l'$. 

    Let $a_0,b_1$ be the endpoints of a $l'$ shortest alternating path $a_0,b_1,...,b_{l'},a_{l'}$ in $V^*$. Let $\mathcal{Z}$ be the collection of possible subsets of $(\Reach(a_0) \cup \Reach(a_{l'})) \cap X$ such that there exists a vertex $v$ that reaches the vertices of the subset and some vertex $b_i$ in the alternating path. That is, $\mathcal{Z} = \{ Z \subseteq (\Reach(a_0) \cup \Reach(a_{l'})) \cap X | \exists v \in V^* \exists i \in [1,l'], Z \cup b_i \subseteq \Reach(v) \}$.

    If $(\Reach(a_0) \cup \Reach(a_{l'})) \cap X \in \mathcal{Z}$, then we are done, as there is a vertex that reaches $(\Reach(a_0) \cup \Reach(a_{l'})) \cap X$ and a vertex $b_i$ in the alternating path. Otherwise, let $Z' \subseteq (\Reach(a_0) \cup \Reach(a_{l'}) $ be the minimum set such that $Z' \not\in \mathcal{Z}$. Note that $|Z'|$, as any individual vertex in $(\Reach(a_0) \cup \Reach(a_{l'})$ will be reachable by either $a_0$ or $a_{l'}$. We distinguish two cases:

    \begin{itemize}
        \item $|Z'| = 2$. Let $Z' = \{u_0,u_{l'}\} \subset X$. Note that one of them must be reached by $a_0$ and the other by $a_{l'}$, otherwise we would have $Z' \in \mathcal{Z}$. Without loss of generality, let $u_0 \in \Reach(a_0)$ and $u_{l'} \in \Reach(a_{l'})$. 
        
        Note that no vertex $b_i$ with $1 \leq i < l'$ can be reached by a vertex $v$ that reaches $u_{l'}$, as in that case we would have that $v$ and $a_0$ form an $l$-connecting path with $l<l'$, therefore we could use the inductive hypothesis to show that there is a vertex $v'$ with $(\Reach(v_0) \cup \Reach(v) \cup \{b_{i'}\})\cap X \subseteq \Reach(v')$, for some $b_{i'}$ in the alternating path, but this would mean that $Z \in \mathcal{Z}$, leading to a contradiction. The same argument can be made to show that no vertex $b_i$ with $1 < i \leq l'$ can be reached by a vertex $v$ that reaches $u_{0}$.

        Note that because $u_0,u_{l'} \in X$ there will be a vertex $y\in Y$ that reaches both $u_0$ and $u_{l'}$. Let $Y'=\{y,a_0,...,a_{l'}\}$ and $X' = \{u_{0},u_{l'},b_1,...,b_{l'}\}$. We can show that $X',Y'$ form a \reachcycle{(l'+2)}. 
        \begin{itemize}
            \item No vertex can reach all $u_{0},u_{l'}$, and a vertex in $b_i$, since then we would have $Z' \in \mathcal{Z}$.
            \item No vertex can reach two vertices in $b_1,...,b'_l$ and $u_{0}$, since then a vertex would reach $u_0$ and some $b_i$ with $1 < i \leq l'$.
            \item Similarly, no vertex can reach two vertices in $b_1,...,b'_l$ and $u_{l'}$, as then a vertex would reach $u_{l'}$ and some $b_i$ with $1 \leq i < l'$.
            Finally, No vertex can reach three vertices in $b_1,...,b'_l$, as then there would be a shorter alternating path between $a_0$ and $a_{l'}$.
        \end{itemize}
        Therefore, $X',Y'$ is a \reachcycle{l'+2}, but $\vec{H}$ does not contain any \reachcycle{l'+2}, leading to a contradiction.

        \item $|Z'| > 2$. Let $\{u_0,u_1,u_2\} \subseteq Z'$. Note that all $Z'\setminus \{u_0\},Z'\setminus \{u_1\},Z'\setminus \{u_2\} \in \mathcal{Z}$, which means that there are three vertices $v'_{0},v'_{1},v'_{2}$ such that $v'_{i}$ reaches $Z'\setminus \{u_i\}$ and some $b_j$ in the alternating path, for $i \in \{0,1,2\}$. Consider the shortest alternating path between each pair of $v'_{0},v'_{1},v'_{2}$, note that at least one of the pairs will form an alternating path with length $l< l'$. Without loss of generality, let $v'_{0},v'_{1}$ be such a pair. Using the induction hypothesis we have that there is a vertex $v''$ such that $(\Reach(v'_{0})\cup \Reach(v'_{1}) \cup \{b_j\})\cap X \subseteq \Reach(v'')$ for some $b_j$ in the alternating path. But note $Z' \subseteq (\Reach(v'_{0})\cup \Reach(v'_{1})) \cap X$, and thus $Z' \in \mathcal{Z}$, which leads to a contradiction.
    \end{itemize}
\end{proof}
    We can now use the claim to complete the proof of the lemma. As we said, assume that we have a connected subset of vertices $V$ such that $X \subseteq_{u\in V} \Reach(v)$. Let $S$ be the set of vertices of $V$ that cannot be reached by any other vertex in $V$, let $l= |S|$ and order the vertices $s_1,...,s_l$ of $S$ arbitrarily. We use induction on $l$ to show the existence of a vertex $v$ in $V^*$ with $X\cap\bigcup_{i\leq l}\Reach(s_i) \subseteq \Reach(v)$.

    The base case is trivial as $s_1$ will satisfy the condition for $l=1$. Now we assume that there exists a vertex $v$ in $V^*$ with $X\cap\bigcup_{i < l}\Reach(s_i) \subseteq \Reach(v)$ and prove that there is a vertex $v'$ with $X\cap\bigcup_{i\leq l}\Reach(s_i) \subseteq \Reach(v)$. Using \Clm{pair_vertices} on $v$ and $s_l$ we have that there is vertex $v'$ with $(\Reach(v) \cup \Reach(s_l)) \cap  X \subseteq \Reach(v')$, which means $X \cap \bigcup_{i\leq l}\Reach(s_i) \subseteq \Reach(v')$, completing the induction.

    This means that there is a vertex $v$ that reaches all the vertices reachable by the vertices in $V$, and therefore reaches all the vertices in $X$, but then $Y,X$ would not be a valid \reachcycle{k'}, which leads to a contradiction, completing the proof of the lemma.

\end{proof}

\begin{lemma} \label{lem:simplex_to_seth}
    Let $\vec{H}$ be a directed acyclic orientation of the hypergraph $H$. If $\Skeleton{l}{\vec{H}}$ does not contain a \reachcycle{k} for any $k$, and contains a \reachsimplex{k'} for some $k'\geq 4$, then $H$ contains a hypergraph $H' \in \setHk{l}{k}$ for some $k\leq k'$, as an induced trimmed subhypergraph.
\end{lemma}
\begin{proof}
    Let $\ord: V(H) \to \NN$ be an ordering of the vertices of $H$ agreeing with $\vec{H}$. Let $k$ be the minimum $k$ such that $\Skeleton{l}{\vec{H}}$ contains a \reachsimplex{k}, we have $4\leq k \leq k'$. And let $X,Y$ be the \reachsimplex{k} in $\Skeleton{l}{\vec{H}}$  with minimum $\sum_{i=0}^{k-1} o(x_i)$. 

    Let $V'$ be the set of vertices that reach some vertex in $X$, $V' = \{v \in \Skeleton{l}{\vec{H}} | \exists x_i \in X, x_i \in \Reach_{\Skeleton{l}{\vec{H}}}(v)\}$. Let $H''$ be the trimmed subhypergraph of $H$ induced by $V'$ 
    
    Let $p$ be the number of connected components in $H'$. We can split the vertices of $V'$ into $p$ disjoint sets $V'_0,...,V'_p$ for each of the connected components of $H''$. For every $p$, let $X_p$ be the subset of vertices of $X$ reachable by some vertex in $V'_p$. Note that from \Lem{groups}, $X_p$ must be a strict subset of $X$; therefore, for every $p$ there exists a vertex $x_i \in X$ such that $X_p \subseteq X \setminus \{x_i\}$, note that there might be more than one such vertex. For every $p$ assign a vertex $x_i$ such that $X_p \subseteq X \setminus \{x_i\}$, let $x_p$ denote such vertex. We can now group the vertices of $V'$ into $k$ sets $V_i = \bigcup_{p | x_p = x_i} V'_p$. 
    
    Now, let $H'$ be the trimmed subhypergraph of $H$ induced by $V' \cup X$,  and $E' = E(H')$. We will show that $H' \in \setHk{l}{3}$. We can group the hyperedges $E'$ into $k$ sets. Let $s(e)$ be the earliest vertex in $e$ according to $\ord$. We will have $E_i = \{e \in E' | s(e) \in V_i \}$. Let $H_i= (V_i,E_i)$.

    Note that $\bigcup_{i<k} V_i = V'$, and therefore $X \cup\bigcup_{i<k} V_i = V(H')$. Similarly $\bigcup_{i<k} E_i = E' = E(H')$. We can verify that $H_i$ will be an \connector{l} for $X \setminus \{x_i\}$. $H_i$ is the union of some of the connected components in $H''$ and the vertices in $X\setminus \{x_i\}$. The vertex $y_i \in Y$ reaches all the vertices in $X\setminus \{x_i\}$, therefore, at least one of such connected components must connect to all the vertices in $X\setminus \{x_i\}$, this satisfies condition $3$ of the definition of \connector{l}. All the other connected components will connect to at least one of those vertices, which means $H_i$ is connected, satisfying condition $1$. Finally, let $e$ be an edge in $E_i$ that contains at least $2$ vertices in $X\setminus \{x_i\}$, note that none of these vertices might be the first $l+1$ in the internal ordering of the edge in $\vec{H}$, as otherwise there would be an edge connecting them in $\Skeleton{l}{\vec{H}}$. Note that those $l+1$ vertices must also be included in $V_i$ since they reach some vertices in $X$ and are part of the same connected component as its source, satisfying the second condition.

    Finally, we have that no hyperedge contains vertices in different subsets, as otherwise they would be in the same connected component, which means that they could not be in different subsets. Therefore, we get that $H' \in \setHk{l}{k}$. 
\end{proof}

We can finally prove \Thm{characterize}, which we restate.
\characterize*

\begin{proof}
    Note that $H$ having $\ldtw{\lparam}(H) > 1$ is equivalent to having an orientation $\vec{H}$ such that its $l$-skeleton $\Skeleton{l}{\vec{H}}$ has $\dtw(\vec{H})=1$, which from \Lem{alpha-acyclic} is equivalent to the reachability hypergraph of  $\Skeleton{l}{\vec{H}}$ being $\alpha$-acyclic. Hence, it suffices to show that such an orientation exists if and only if $H$ contains a graph in $\setH{l}$ as an induced trimmed subhypergraph.
    \begin{itemize}
        \item We first prove the ``if'' part, let $H' \in \setHk{l}{k}$ be an induced trimmed subhypergraph of $H$. Let $X$, $V_0,...,V_{k-1}$ be the sets of vertices of $V(H')$ and $E_0,...,E_{k-1}$ the sets of hyperedges that satisfy the conditions of \Def{seth}. We create an ordering $\ord$ of the vertices of $H$ in which first we will have all the vertices in the $V_i$ sets, then the vertices in $X$, and finally the vertices in $V(H)\setminus V(H')$. 
        
        For every set $V_i$, select an arbitrary vertex $v_i$ in one of the connected components that connects to all the vertices in $V\setminus \{v_i\}$, and for every $j\neq i$ set the vertices in the shortest path from $v_i$ to $x_j$ in ascending order, allowing $v_i$ to reach $x_j$ in any $l$-skeleton. We argue that the sets $Y=\{v_0,...,v_{k-1}\}$ and $X$ form a \reachsimplex{k} in $\Skeleton{l}{\vec{H}}$.

        From the way we ordered the vertices we have that $X\setminus \{x_i\} \subseteq \Reach(v_i)$ for all $i$. Only rest to show is that no vertex can reach all the vertices in $X$. Assume there is a vertex $v$ such that $X \subseteq \Reach(v)$, note that $v \not \in V(H) \setminus V(H')$, as the vertices in $V(H) \setminus V(H')$ are all after the vertices in $X$. If $v \in X$ it means that there is a direct edge connecting a pair of vertices in $X$. Otherwise, if $v \in V_i$ for some $i$, then the path to $x_i$ must necessarily include another vertex in $X$, also meaning that there must be a direct edge connecting a pair of vertices in $X$.

        We can prove that no direct edge can connect two vertices in $X$. Let $e$ be any hyperedge in $H$ containing two vertices $x_j,x_{j'} \in X$, and $e' = e \cap V(H')$, we know $|e'|>2$ and there must be an $i \neq j,j'$ such that $e'\in E_i$. But $H_i=(V_i\cup X\setminus\{x_i\},E_i)$ is an \connector{l} of the set $X\setminus \{x_i\}$ which means that $e'$ must contain at least a vertex $l+1$ vertices in $V_i$, and therefore the $l$-skeleton of $\vec{H}$ will not include a direct edge between $x_i$ and $x_j$.

        \item Now we prove the only if part, let $\vec{H}$ be an orientation of $H$ such that the reachability graph of $\Skeleton{l}{\vec{H}}$ is $\alpha$-acyclic. We show that $H$ contains a hypergraph in $\setH{l}$ as an induced trimmed subhypergraph. From \Thm{acyclic} we have two cases, the first one is equivalent to $\Skeleton{l}{\vec{H}}$ containing a \reachcycle{k} for some $k\geq 3$ and the second is equivalent to $\Skeleton{l}{\vec{H}}$ containing a \reachsimplex{k} for some $k\geq 3$. In the first case, from \Lem{cycle_to_seth} we find that $H$ contains a pattern in $\setHk{l}{3}$ as an induced trimmed subhypergraph. In the second case, using \Lem{simplex_to_seth} we find that $H$ contains a pattern in $\setHk{l}{k'}$ for some $k'\leq k$.
    \end{itemize}
\end{proof}

\subsection{The special case of $\setH{\infty}$}

As mentioned in the introduction, for the case of $\lparam =\infty$, suffices to look at which patterns have an $\licl$ greater or equal than $6$.

\linfty*
\begin{proof}
	We prove that $H$ will contain a pattern $H' \in \setH{\infty}$ as an induced trimmed subhypergraph if and only if $\licl(H^c)\geq 6$. 
	
	We first prove the if direction. Let $H$ be a pattern with $\licl(H^c)\geq 6$. Let $t$ be the $\licl$, and $V'=v_1,...,v_t$ the vertices in the cycle in the order. Let $H'$ be the trimmed subhypergraph induced by $V'$, clearly it will only contain the edges of the cycle and $H'$ is an induced cycle. We can set $C={c_1=v_1,c_2=v_3,c_3=v_5}$, $D_1 =\{v_4\}$,$D_2=\{v_6,...,v_t\}$ and $D_3 = \{v_2\}$ and see that $H' \in \setH{\infty}$.
	
	We now prove the only if part. Let $H$ be a pattern containing $H' \in \setH{\infty}$, $H^c$ the clique completion of $H$ and $H''$ the clique completion of $H'$. Consider the vertices $c_1,c_2,c_3$ in the core, and consider the sets $D_1,D_2,D_3$. Find the shortest paths between $c_1$ and $c_2$, $c_2$ and $c_3$, and $c_3$ and $c_1$, in $D_1$,$D_2$ and $D_3$. These paths must have length at least $2$ and do not share edges with each other. Therefore, forming an induced cycle of length at least $6$. With implies $\licl(H^c) \geq \licl(H'') \geq 6$. 
\end{proof}

\section{Hardness of counting hypergraph homomorphisms} \label{sec:hardness}

We prove hardness using colorful subgraph counts. Given a graph, a coloring is a function $c: V(G) \to [h]$. A homomorphism $\phi$ is colorful if it maps every vertex to a vertex of a different color. This implies that the homomorphism must be injective, as it can not map to the same vertex twice. We use $\ColHom{G}{H}$ for the number of colorful homomorphisms from $H$ to $G$. We can show that like in graphs (\cite{BeGiLe+22,Me16}), one can reduce counting colorful hypergraph homomorphisms to uncolored homomorphisms.

\begin{lemma} \label{lem:colorful_hom}
    Let $\lparam \in \mathbb{Z}^+ \cup \{\infty\}$ and $H$ be a hypergraph. If there is an $O(n^\gamma)f(\ldegen{l}(G))$ algorithm for computing $\Hom{G}{H}$ for all hypergraphs $G$, then there is an $O(n^\gamma)f(\ldegen{l}(G))$ algorithm for computing $\ColHom{G}{H}$ for any colored hypergraph $G$ with $h=|V(H)|$ colors.
\end{lemma}
\begin{proof}
    Fix $\lparam$. Let $G$ be a colored hypergraph with colors $c: V(G) \to [h]$, let $n$ be the number of vertices of $G$ and $\ldegen{l} = \ldegen{l}(G)$ its $l$-degeneracy. Let $I\subset[h]$ be a subset of the colors; we can define $V_I$ as the set of vertices of $G$ with color in $I$, $V_I = \{v \in V(G) | c(v) \in I\}$.

    Let $G[V_I]$ be the subhypergraph of $G$ induced by the vertices of $V_I$. The number of homomorphisms $\phi$ from $H$ to $G$ such that $c(\img(\phi))$ is  equal to $\Hom{G[V_I]}{H}$. Note that for that to hold we need to use the induced subhypergraphs and not the induced \emph{trimmed} subhypergraphs.

    A colorful homomorphism $\phi$ will require that $c(\img(\phi)) = [h]$. We can use inclusion-exclusion to find the number of colorful homomorphisms:

    \[
        \ColHom{H}{G} = \sum_{I \subseteq [h]} (-1)^{|I|}\Hom{G[I]}{H}
    \]
    Note that for all graphs $G[I]$, we have $V(G[I]) \leq n$ and $\ldegen{l}(G[I]) < \ldegen{l}$. Therefore, for each $I \subseteq [h]$ we can compute $\Hom{G[I]}{H}$ in time $O(n)f(\ldegen{l}(G))$. Assuming $h = O(1)$, we can then use all the counts to compute $\ColHom{H}{G}$ in total time $O(n)f(\ldegen{l}(G))$.
\end{proof}

We also show that, for any $k$, detecting a $k$-simplex can be done in the same time as detecting a colorful $k$-simplex.

\begin{lemma} \label{lem:colorful_simplex}
    Let $k=O(1)$. If there is an algorithm for detecting a colorful $k$-simplex in a colored hypergraph with $m$ hyperedges in time $f(m)$, then there is an algorithm for detecting a $k$-simplex in an uncolored hypergraph with $m$ hyperedges in time $\tilde{O}(f(m))$.
\end{lemma}
\begin{proof}
    Given a graph $G$ we can randomly color $G$ into $G'$ using the the $k+1$ colors of the colorful $k$-simplex uniformly at random. If $G$ contained a $k$-simplex, we will have that $G'$ contains a colorful $k$-simplex with constant probability $(1/k)$, we can then use the $f(m)$ algorithm. This process can be derandomized using color-coding \cite{AlYuZw94}, adding a logarithmic factor.
\end{proof}

The last piece of the reduction is the construction of the graph $G^H_l$. This reduction is based on the construction shown in \cite{BePaSe21}. The idea is to take a colored regular arity hypergraph $G$ and replace every hyperedge with part of the pattern hypergraph $H$. Obtaining a hypergraph such that a colorful copy $H$ exists if and only if there is a colorful simplex in $G$.

For a hypergraph $H$ containing $H' \in \setHk{l}{k}$ as an induced trimmed subhypergraph, we will use $V^\{ext\}$ for the set of vertices $V(H) \setminus V(H')$, and $E^\{ext\} = \{e \in E(H) : e \subseteq V^{ext}\}$ for the set of edges contained entirely in $V^{ext}$. We also use $X$ to refer to the core vertices of $H'$, and $V_i$ for each of the partitions of vertices of $H'$ that satisfy the properties of \Def{seth}. We set $E_{i} = \{e \in E(H)\setminus E^{ext} | e \subseteq X\cap V_i \cap V^{ext} \}$.

\begin{definition} [$G^H_l$] \label{def:ghl}
    Let $H$ be a hypergraph pattern with $h$ vertices containing $H' \in \setHk{l}{k}$ as an induced trimmed subhypergraph. Let $G$ be an $k$-colored regular $(k-1)$-regular arity hypergraph. Let $c: V(G) \to [k]$ be the coloring function of $G$. And let $\ord: V(H) \to [h]$ be an ordering of the vertices of $H$, such that for all $x_i \in X$, $\ord(x_i) = i$. We construct the $h$-colored hypergraph $G^H_l$ as follows.
    
    \begin{itemize}
        \item Initialize $G^H_l$ as an empty hypergraph with coloring function $c'$.
        \item For all $v \in V(G)$, add $v$ to $V(G^H_l)$ with $c'(v) = c(v)$.
        \item For all $v \in V^{ext}$, add $v$ to $V(G^H_l)$ with $c'(v) = \ord(v)$. Add each hyperedge $e \in E^{ext}$ to $E(G^H_l)$ connecting the corresponding vertices.
        \item For each colorful hyperedge $e \in E(G)$, that is, a hyperedge with $k-1$ distinct colors in its vertices, let $i$ be the color not in $e$, and $y_j \in e$ be the vertex in $e$ with color $j \neq i$. Add a copy $V^e_i$ of $V_i$ to $V(G^H_l)$ with $c'(v) = \ord(v)$ for every $v\in V_i$.
        For every hyperedge $e' \in E_i$ we will create a hyperedge $e''$ in $E(G^H_l)$ such that:
        \begin{itemize}
            \item $\forall x_j \in X \cap e'$, we add $y_j$ to $e''$.
            \item $\forall v \in V_i\cap e'$, we add $u \in V^e_i$ with $c(u)=o(v)$ to $e''$.
            \item $\forall v \in V^{ext}$, we add $u \in V(G^H_l)$ with $c(u) = o(v)$ to $e''$.
        \end{itemize}
    \end{itemize}
\end{definition}

\begin{lemma} \label{lem:construct_hgl}
    $G^H_l$ can be constructed in $O(n+m)$ time, where $m$ is the number of hyperedges of $G$. Moreover, $G^H_l$ will have $\ldegen{l}(G^H_l) = O(1)$.
\end{lemma}
\begin{proof}
    First, we will create a vertex $v$ for each vertex in $V(G)$, which will take $O(n)$ time. Then, for each hyperedge we need to add a subhypergraph of $H$ to $G^H_l$, this will take $O(m)$ total time, as the number of vertices and hyperedges of $H$ is constant with respect to the size of $G$. Therefore, it will take $O(n+m)$ total time.
    
    Now, we show that $G^H_l$ has constant $l$-degeneracy. We give an ordering $\ord: V(G^H_l) \to [|V(G^H_l)|]$ of the vertices such that for the directed hypergraph $\vec{G}^H_l$ obtained from applying $\ord$ to $G^H_l$ we have $\Delta^+_l(\vec{G}^H_l) = O(1)$.
    
    We divide the vertices of $V(G^H_l)$ in three portions: $V^{(3)}$ for the vertices corresponding to $V^{ext}$ in $H$, $V^{(2)}$ for the vertices corresponding to $V(G)$ and $V^{(1)}$ for the remaining vertices corresponding with vertices in the portions $V_i$ of $H$. Let $\ord$ be any ordering such that all vertices in $V^{(1)}$ precede all vertices in $V^{(2)}$, and all vertices in $V^{(2)}$ precede all vertices in $V^{(3)}$. We can show that the $l$-outdegree of each vertex in the directed graph will be constant:
    \begin{itemize}
    	\item If $u \in V^{(1)}$, then $u$ can only be included in at most $|E(H)| = O(1)$ hyperedges in $G^H_l$. The arity of each hyperedge is at most $V(H)=O(1)$, and therefore $u$ will only have a constant number of neighbors.
    	\item If $u \in V^{(2)}$, note that every hyperedge that includes two vertices in $V^{(2)}$ must contain at least $l+1$ vertices in $V^{(1)}$, therefore, the $l$-skeleton of $\vec{G}^H_l$ can not include a direct edge between two vertices in $V^{(2)}$, thus, $u$ can only have as $l$-out-neighbors vertices in $V^{(3)}$. Which implies $d^+(u) = O(1)$.
    	\item If $u \in V^{(3)}$, then $d^+(u) < |V^{(3)}| = O(1)$.
	\end{itemize}
\end{proof}

\begin{lemma} \label{lem:equivalence_simplex_to_pattern}
    $G^H_l$ contains a colorful copy of $H$, if and only if $G$ contains a colorful $k$-simplex.
\end{lemma}
\begin{proof}
    We first prove the if part. Let the vertices $v_0,...,v_{k-1}$ of $G$ induce a copy of a colorful $k$-simplex in $G$, we show how to find a colorful copy of $H$ in $G^H_l$.
    Let $e_i = \{v_0,...,v_{k-1} \}\setminus \{v_i\}$. Note that $e_i \in E(G)$, otherwise $v_0,...,v_{k-1}$ do not induce a simplex.
	Consider the sets $V^{e_i}_i$ for all $i$, and $V^{ext}$ in $G^H_l$. We show that $\{v_0,...,v_{k-1}\} \cup \bigcup_i V^{e_i}_i \cup V^{ext}$ induce a colorful copy of $H$. First, note that every vertex has a different color, as every vertex in $H$ corresponds to one vertex in the set. Moreover, every edge in $H$ is either in $E^{ext}$ or in one of the sets $E_i$, in both cases from the construction it will appear containing the corresponding vertices in $v_0,...,v_{k-1} \cup \bigcup_i V^{e_i}_i \cup V^{ext}$.
    
    We now prove the only if part, let the set $V'$ induce a colorful copy of $H$ in $G^H_l$, $V'$ must contain vertices $y_0,...,y_{k-1}$ with colors $0$ to $k-1$ that are also present in $V(G)$. Moreover, for each subset $\{y_0,...,y_{k-1} \} \setminus \{y_i\}$ we will have that $G$ must contain an edge $e_i$ containing them, otherwise the set $E_i$ cannot be completely included in the subhypergraph induced by $V'$. Therefore, it must be the case that $\{y_0,...,y_{k-1} \}$ induce a colorful $k$-simplex.
\end{proof}

\hardness*
\begin{proof}
	From \Lem{colorful_hom} we will have that the $f(\ldegen{l})O(n^\gamma)$ algorithm for computing $\Hom{G}{H}$ implies an $f(\ldegen{l})O(n^\gamma)$ algorithm for computing $\ColHom{G}{H}$.
	
	We construct $G^H_l$ from $G$ as in \Def{ghl}, from \Lem{construct_hgl} we have that $\ldegen{l}(G^H_l) = O(1)$, moreover $G^H_l$ will have $O(m)$ hyperedges. Therefore, we can count the number of colorful homomorphisms of $H$ in $G^H_l$ in time $O(m^\gamma)$. 
	
	From \Lem{equivalence_simplex_to_pattern} we have that $\ColHom{G^H_l}{H} > 0$ if and only if $G$ contains a colorful $k$-simplex, effectively giving an $O(m^\gamma)$ algorithm for colorful $k$-simplex detection, which itself implies an $\tilde{O}(m^\gamma)$ algorithm for $k$-simplex detection using \Lem{colorful_simplex}.
\end{proof}

\section{From homomorphisms to subhypergraphs} \label{sec:sub}

In this section, we show how to obtain a dichotomy for linear time subhypergraph counting from the hypergraph homomorphism counting dichotomy. Given by \Thm{sub}.

In a recent work, Bressan et al. showed that, similar to what happened in graphs, one can write the number of subhypergraph as a linear combination of hypergraph homomorphism \cite{BrBrDe+25}. Moreover, they generalized the result of \cite{CuDeMa17}, showing that we can extend hardness from the homomorphisms into the subhypergraph (or any hypergraph motif parameter). We will show that, similar to the case of graphs \cite{BeGiLe+22}, we can also use the hypergraph homomorphism basis to extend more fine-grained complexity lower bounds from homomorphisms into subhypergraphs.

\Lem{hombasis} together with \Thm{main} is enough to prove the upper bound of \Thm{sub}. The lower bound requires a bit more of care and generalizes some of the results in \cite{BeGiLe+22}. We will prove the following lemma, which is analogous to Lemma 4.2 in \cite{BeGiLe+22}.

\begin{lemma} \label{lem:sub_to_hom}
	Let $H_1,...,H_k$ be pairwise non-isomorphic hypergraphs and let $c_1,...,c_k$ be non-zero constants. For every hypergraph $G$ there are hypergraphs $G_1,...,G_k$ computable in time $O(|V(G)+|E(G)|)$, such that for every $i \in [k]$, $|V(G_i)| = O(V(G))$, $|E(G_i)| = O(E(G))$ and $\ldegen{l}(G_i) = O(\ldegen{l}(G))$ for all $l\geq 0$; and such that knowing $b_j = \sum_{i \in [k]} c_i \Hom{G_j}{H_i}$ for all $j \in [k]$, allows one to compute $\Hom{G}{H_i}$ for all $i \in [k]$ in time $O(1)$.
\end{lemma}

To prove this lemma we require of a tensor product for hypergraphs that preserve homomorphisms and does not increase the $l$-degeneracy. Fortunately, we can show that the tensor product introduced in \cite{BrBrDe+25} satisfies all these properties.

Given two hypergraphs $H$ and $G$, one can define the $G$-projection and $H$-projections of their Cartesian products as the functions $\xi_G: V(G) \times V(H) \to V(G)$ and $\xi_G: V(G) \times V(H) \to V(H)$ that map each pair in $V(G) \times V(H)$ to its first or second component respectively. For a set $S \subseteq  V(G)\times V(G)$ we can define:

\begin{equation*}
	\xi_G(S) = \{\xi_G(s) : s \in S\} \ \ \ \ 	\xi_H(S) = \{\xi_H(s) : s \in S\}
\end{equation*}

We can now define the tensor product:
\begin{definition}[Tensor product of hypergraphs \cite{BrBrDe+25}] \label{def:tensor}
	The tensor product of two hypergraphs $G \otimes H$ is the hypergraph with vertex set $V(G) \times V(H)$ and hyperedge set $\{e \subseteq V(H)\times V(H) : \xi_G(e) \in E(G) $ and $ \xi_H(e) \in E(H) \}$.
\end{definition}

Bressan et al. also showed that the homomorphisms are preserved under the tensor product:

\begin{lemma} \cite{BrBrDe+25} \label{lem:homtensor}
	All hypergraphs $F,H,G$ satisfy:
	\[
		\Hom{G\otimes H}{F} = \Hom{G}{F} \cdot \Hom{H}{F}
	\]
\end{lemma}

We can show that if $H$ is constant sized, the product $G \otimes H$ will maintain the $l$-degeneracy of $G$.

\begin{lemma} \label{lem:sizes}
	Let $G$ and $H$ be two hypergraphs with $|V(H)| = O(1)$ and bounded rank. It holds:
	\begin{itemize}
		\item $|V(G \otimes H)| = O(V(G))$.
		\item $|E(G \otimes H)| = O(E(G))$.
		\item For all $l\geq 0$, $\ldegen{l}(G \otimes H) = O( \ldegen{l}(G))$.
	\end{itemize}
\end{lemma}
\begin{proof}
	We proof each part separately:
	\begin{itemize}	
	\item The vertex set of $G \otimes H$ is $V(G) \times V(H)$ which will have size $|V(G)| \cdot |V(H)| = O(|V(G)|)$.
	
	\item For every hyperedge $e \in G$ we can have multiple hyperedges $e' \in G \otimes H$ with $\xi_G(e') = e$, however, we can show that the set of such hyperedges is bounded. Let $e'$ be one such hyperedge, it must contain, for each vertex $v \in e$ a subset of all the vertices $v \times V(H)$ in $ G \otimes H$. The number of possible subsets is $2^{|V(H)|} = O(1)$ and we have that $|e| = O(1)$ as $G$ has bounded rank. Therefore, at most a constant number of hyperedges can have $\xi_G(e') = e$. Which implies $|E(G \otimes H)| = O(E(G))$.
	
	\item Finally, fix some $l$, let $\ord: V(G) \to [V(G)]$ be an ordering of the vertices of $G$ that agrees with the $l$-degeneracy orientation of $G$, $\vec{G}^{(l)}$. Construct an ordering $\ord': V(G \otimes H) \to [G \otimes H]$ such that for every $u,v \in G \otimes H$, $\ord'(u) < \ord'(v)$ if $\ord(\xi_G(u)) < \ord(\xi_G(u))$. let $\vec{G \otimes H}$ be the result of orienting $G \otimes H$ according to $\ord'$. We show that $\Delta^+_l(\vec{G \otimes H}) = O(\ldegen{l}(G))$, which implies $\ldegen{l}(G \otimes H) = O(\ldegen{l}(G))$.
	
	Let $u \in V(G \otimes H)$ be a vertex with $\xi_G(u) = v$, let $d$ be the $l$-outdegree of $v$ in $\vec{G}$, from \Lem{degeneracy_ordering} we have $d = O(\ldegen{l}(G))$, let $N^+_l(v)$ be the $l$-out-neighborhood of $v$. We can show that every $l$-out-neighbor $u'$ of $u$ must have $\xi_G(u') = v$ or $\xi_G(u') \in N^+_l(v)$. 
	
	Consider otherwise, there exists a vertex $u'$ with $\xi_G(u') = v' \not\in N^+_l(v)\cap \{v\}$ that is an $l$-out-neighbor of $u$. There must be a hyperedge $e \in \vec{G \otimes H}$ such that $\{u,u'\} \subseteq e$, such that $u$ is one of the first $l+1$ vertices in the ordering. Now consider the hyperedge $e' \in E(G)$ such that $\xi_G(e) = e'$, it must include both $v$ and $v'$, moreover the position of $v$ in $e'$ can not be later than the position of $u$ in $e$, therefore we will have that $v'$ must be an $l$-out-neighbor of $v$ in $G$, reaching to a contradiction.
	
	Therefore, the $l$-outdegree of $u$ will be $d^+_l(u) \leq V(H) \cdot (|N^+_l(v)|+1) \leq V(H) \cdot (d+1) = O(\ldegen{l}(G))$.
	\end{itemize}
\end{proof}

The last piece we need in order to prove \Lem{sub_to_hom} is to generalize a lemma from Erd{\H{o}}s, Lov{\'a}sz and Spencer to the hypergraph setting. This lemma was originally shown in \cite{ErLoSp78}, also appears as Proposition $5.44(b)$ in \cite{Lo12}. We restate (and extend to hypergraphs) as in Lemma $A.2$ from \cite{BeGiLe+22}.

\begin{lemma} \label{lem:lovasz}
	Let $H_1,...H_k$ be pairwise non-isomorphic hypergraphs, and let $c_1,...,c_k$ be non-zero constants. Then there exists hypergraphs $F_1,...F_k$ such that the $k \times k$ matrix $M_{i,j} = c_j \Hom{F_i}{H_j}, 1\leq i,j \leq k$, is invertible.
\end{lemma}
\begin{proof}
	The proof is similar to the one in \cite{Lo12}. We construct $F'_i$ by weighting the vertices of each $H_i$ with algebraically independent weights. $\Hom{F'_i}{H_i}$ will represent the sum of weighted hypergraph homomorphisms, that is:
	\[
	\Hom{F'_i}{H_i} = \sum_{\phi \in \Phi(H_i,F'_i)} \prod_{u \in H_i} w(\phi(u))
	\]
	Where $w$ represent the weight of the corresponding vertex in $F'_i$.
	
	We can show that the matrix $[\Hom{F'_j}{H_i}]_{i,j=1}^k$ has non-zero determinant. If every weight is considered a separate variable, the determinant of the matrix can be written as a polynomial, the multi-linear part of the polynomial must correspond to injective homomorphisms. We can show that the determinant of the injective part of the matrix is non-zero, which implies that the determinant will also be non-zero.
	
	Note that two hypergraphs can map to each other injectively if and only if they are isomorphic, moreover if $H$ has an injective homomorphism to $H'$ and $H'$ does to $H''$, then $H$ will have to $H''$. Therefore, one can reorder the hypergraphs $F'_i$ such that the matrix of injective homomorphisms is upper triangular with non-zero entries in the diagonal. Therefore its determinant will be non-zero.
	
	 The polynomial can not become $0$ for all integers and therefore we can replace the algebraic independent weights with positive integer to get an invertible matrix $[\Hom{F'_j}{H_i}]_{i,j=1}^k$.
	
	We now construct hypergraphs $F_i$ by replacing every vertex $v$ in $F'_i$ by $w(v)$ copies of it, $v_1,...,v_{w(v)}$, and for each hyperedge $e \in F'_i$, replacing it with a total of $\prod_{v\in e} w(v)$ hyperedges, each being a set in $\bigtimes_{v \in e} \{v_1,...,v_{w(v)} \}$. We can show that $\Hom{F_i}{H_i} = \Hom{F'_i}{H_i}$ and therefore the matrix $[\Hom{F_j}{H_i}]_{i,j=1}^k$ is invertible.
	
	Consider a homomorphism $\phi'$ from $H_i$ to $F'_j$, we can show that it corresponds $\prod_{u \in H_i} w(\phi(u))$ on $F_j$. We can construct $\phi$ by replacing $\phi'(u)$ for any of the $w(\phi'(u))$ vertices created from it. If $\phi'(e) \in E(F'_j)$ then we will have that $\phi(e) \in E(F_j)$ and therefore it will be a valid homomorphism. Similarly, we can construct a homomorphism $\phi'$ from $H_i$ to $F'_j$ from a homomorphism $\phi'$ from $H_i$ to $F_j$ by setting $\phi'(u)$ to be the original vertex from $\phi(u)$. Therefore we will have:
	\[
		\Hom{F'_i}{H_i} = \sum_{\phi \in \Phi(H_i,F'_i)} \prod_{u \in H_i} w(\phi(u)) = \Hom{F_i}{H_i}
	\]
\end{proof}

We can finally prove \Lem{sub_to_hom}.
\begin{proof}[Proof of Lemma \ref{lem:sub_to_hom}]
	The proof is similar to the prof of Lemma $4.2$ in \cite{BeGiLe+22}.
	
	Using \Lem{lovasz}, we have that there exists hypergraphs $F_1,...,F_k$ such that the matrix $M_{i,j} = c_j \Hom{F_i}{H_j}, 1\leq i,j \leq k$ is invertible. For an input hypergraph $G$, for every $1\leq j \leq k$ we can set $G_j = G \otimes F_j$ and $b_j =  \sum_{i =1}^k c_i \Hom{G_j}{H_i}$.
	
	We will have:
	\begin{equation} \label{eq:bj}
		b_j =  \sum_{i =1}^k c_i \Hom{G_j}{H_i} = \sum_{i =1}^k  c_i \Hom{F_i}{H_i} \Hom{G}{H_i} = \sum_{i =1}^k M_{j,i}\Hom{G}{H_i}
	\end{equation}

	Where the second equality comes from \Lem{homtensor} and third from the definition of $M_{i,j}$.
	
	We can build a system of $k$ linear equations using the expression of \Eqn{bj} with $\Hom{G}{H_i}$ as the $k$ variables, $M$ as the matrix of the system and $b_j$ as the constant terms. Because $M$ is invertible, if we know all the values $b_j$, then we can compute $\Hom{G}{H_i}$ in constant time.
	
	Note also that every hypergraph $F_i$ has $|V(F_i)| =O(1)$, and therefore using \Lem{sizes} we will get that all the hypergraphs $G_j$, will have $V(G_j) = O(V(G))$, $E(G_j) = O(E(G))$ and for all $l\geq 0$, $\ldegen{l}(G_j) = O(\ldegen{l}(G))$.
	
	Moreover, we can compute each $F_i$ in time $|V(G)||V(F)|+|E(G)|2^{|V(H)|}r(G) = O(|V(G)|+|E(G)|)$.
	\end{proof}
	
	We can now prove the following theorem which implies the lower bound of \Thm{sub}.
	\begin{theorem}
		Let $H$ be a pattern and $\epsilon > 1$, if there is an $f(\ldegen{l}(G))O(n^\epsilon)$ algorithm for computing $\Sub{G}{H}$ for all inputs $G$, then, for any pattern $H' \in \cQ(H)$ in the quotient set of $H$, we can compute $\Hom{G}{H'}$ in $f(\ldegen{l}(G))O(n^{\epsilon})$ time.
	\end{theorem}
	\begin{proof}
		From \Lem{hombasis} we can express $\Sub{G}{H}$ as a linear combination of homomorphism counts for patterns in the quotation set of $H$.
		\[
			\Sub{G}{H} = \sum_{H' \in \cQ(H)} \gamma(H') \Hom{G}{H'}
		\]
		Let $k = |\cQ(H)|$ and set $\{H_1,...,H_k\} = \cQ(H)$ and $c_i = \gamma(H_i)$. From \Lem{hombasis} we have that all $c_i \neq 0$.
		
		Given $G$, we create the hypergraphs $G_1,...,G_k$ as in \Lem{sub_to_hom}. Because the size and $l$-degeneracy of these hypergraphs is the same than $G$ we can use the $f(\ldegen{l}(G))O(n^\epsilon)$ to compute $\Sub{G_j}{H}$.
		
		Then note that for each $j$:
		\[
			b_j = \sum_{i=1}^k c_i \Hom{G_j}{H_i} = \sum_{H' \in \cQ(H)} \gamma(H') \Hom{G_j}{H'} = \Sub{G_j}{H}
		\]
		Therefore we can compute all the values in $b_j$ in $f(\ldegen{l}(G))O(n^\epsilon)$ time, which in turn allows us to compute $\Hom{G}{H'}$ for all $H' \in \cQ(H)$ in additional constant time using \Lem{sub_to_hom}.
	\end{proof}

\bibliographystyle{alpha}
\bibliography{subgraph_counting_doi,hypergraphs}

\end{document}